\documentclass{article}

\usepackage{microtype}
\usepackage{graphicx}
\usepackage{subfigure}
\usepackage{booktabs} 

\usepackage{hyperref}

\usepackage{amssymb}
\usepackage{amsmath}
\usepackage{tikz}
\usetikzlibrary{calc}
\usepackage{todonotes}
\usepackage{dsfont}
\usepackage{enumitem}
\usepackage{tikz}
\usepackage{amsthm}
\usepackage{natbib}
\usepackage{thm-restate}
\usepackage{multirow}
\usepackage{amsfonts}
\usepackage{mathtools}
\usepackage{wrapfig}
\usepackage{bm}
\usepackage{cleveref}
\usepackage{mleftright}

\newtheorem{theorem}{Theorem}
\newtheorem{corollary}{Corollary}[theorem]
\newtheorem{lemma}[theorem]{Lemma}
\newtheorem{definition}{Definition}

\newcommand{\bb}[1]{\bm{#1}}
\newcommand*\circled[1]{\tikz[baseline=(char.base)]{
		\node[shape=circle,draw,inner sep=1pt] (char) {\scriptsize #1};}}

\newcommand{\NPHard}{$\mathsf{NP}$-hard}
\newcommand{\NP}{$\mathsf{NP}$}
\newcommand{\RP}{$\mathsf{RP}$}

\newcommand{\svec}{\bb{s}}
\newcommand{\sset}{\mathcal{S}}
\newcommand{\Expec}{\mathbb{E}}
\newcommand{\defeq}{\vcentcolon=}

\newcommand{\opt}{\textsc{Opt}}
\newcommand{\apx}{\textsc{Apx}}

\newcommand{\Scal}{\mathcal{S}}

\newcommand{\Hcal}{\mathcal{H}}
\newcommand{\kvec}{\bb{k}}

\newcommand{\ind}{	\mathsf{1}}
\newcommand{\Ical}{\mathcal{I}}
\newcommand{\rec}{\mathcal{R}}
\newcommand{\A}{\mathcal{A}}

\newcommand{\zvec}{\bb{z}}

\newcommand{\nState}{d}

\newcommand{\event}{\mathbb{I}}

\newcommand{\M}{\mathcal{M}}
\newcommand{\G}{\mathcal{G}}
\newcommand{\B}{\mathcal{B}}
\newcommand{\Y}{\mathcal{Y}}
\newcommand{\E}{\mathcal{E}}
\newcommand{\xvec}{\bb{x}}
\newcommand{\yvec}{\bb{y}}

\newcommand{\wvec}{\bb{w}}
\newcommand{\vvec}{\bb{v}}
\newcommand{\dvec}{\bb{d}}

\newcommand{\usend}{u^\mathsf{s}}

\newcommand{\K}{\mathcal{K}}

\newcommand{\Ocal}{\mathcal{O}}
\newcommand{\muvec}{\boldsymbol{\mu}}
\newcommand{\lambdavec}{\boldsymbol{\lambda}}
\newcommand{\X}{\mathcal{X}}

\usepackage[accepted]{icml2021}

\icmltitlerunning{Multi-Receiver Online Bayesian Persuasion}

\begin{document}

\twocolumn[
\icmltitle{Multi-Receiver Online Bayesian Persuasion}



\icmlsetsymbol{equal}{*}

\begin{icmlauthorlist}
\icmlauthor{Matteo Castiglioni}{polimi}
\icmlauthor{Alberto Marchesi}{polimi}
\icmlauthor{Andrea Celli}{polimi}
\icmlauthor{Nicola Gatti}{polimi}
\end{icmlauthorlist}

\icmlaffiliation{polimi}{Politecnico di Milano, Milan, Italy}
\icmlcorrespondingauthor{Matteo Castiglioni}{matteo.castiglioni@polimi.it}

\icmlkeywords{Bayesian persuasion, multi-receiver, online learning}

\vskip 0.3in
]



\printAffiliationsAndNotice{}  

\begin{abstract}
	Bayesian persuasion studies how an informed sender should partially disclose information to influence the behavior of a self-interested receiver.
	Classical models make the stringent assumption that the sender knows the receiver's utility.
	This can be relaxed by considering an \emph{online learning} framework in which the sender repeatedly faces a receiver of an unknown, adversarially selected type.
	We study, for the first time, an online Bayesian persuasion setting with \emph{multiple receivers}.
	We focus on the case with \emph{no externalities} and binary actions, as customary in offline models.
	Our goal is to design \emph{no-regret} algorithms for the sender with polynomial per-iteration running time.
	First, we prove a negative result: for any $0 < \alpha \leq 1$, there is no polynomial-time no-$\alpha$-regret algorithm when the sender's utility function is \emph{supermodular} or \emph{anonymous}.
	%
	Then, we focus on the case of \emph{submodular} sender's utility functions and we show that, in this case, it is possible to design a polynomial-time no-$\left( 1 - \frac{1}{e} \right)$-regret algorithm.
	To do so, we introduce a general \emph{online gradient descent} scheme to handle online learning problems with a finite number of possible loss functions.
	This requires the existence of an approximate projection oracle. We show that, in our setting, there exists one such projection oracle which can be implemented in polynomial time.
\end{abstract}

\section{Introduction}\label{sec:introduction}

Bayesian persuasion was originally introduced by~\citet{kamenica2011bayesian} to model multi-agent settings where an informed \emph{sender} tries to influence the behavior of a self-interested \emph{receiver} through the strategic provision of payoff-relevant information.
Agents' payoffs are determined by the receiver's action and some exogenous parameters collectively termed the \emph{state of nature}, whose value is drawn from a common prior distribution and observed by the sender only.
Then, the sender decides how much of her/his private information has to be revealed to the receiver, according to a public randomized policy known as \emph{signaling scheme}. 
From the sender's perspective, this begets a decision-making problem that is essentially about controlling ``who gets to know what''.
This kind of problems are ubiquitous in application domains such as auctions and online advertising~\citep{bro2012send,emek2014signaling,badanidiyuru2018targeting}, voting~\citep{alonso2016persuading,cheng2015mixture,castiglioni2019persuading,castiglioni2020persuading}, traffic routing~\citep{vasserman2015implementing,bhaskar2016hardness,castiglioni2020signaling}, recommendation systems~\citep{mansour2016bayesian}, security~\citep{rabinovich2015information,xu2016signaling}, and product marketing~\citep{babichenko2017algorithmic,candogan2019persuasion}.\footnote{Persuasion was famously attributed to a quarter of the GDP in the United States by~\citet{mccloskey1995one}, with a more recent estimate placing this figure at $30\%$~\citep{antioch2013persuasion}.}

The classical Bayesian persuasion model by~\citet{kamenica2011bayesian} makes the stringent assumption that the sender knows the receiver's utility exactly.
This is unreasonable in practice.
Recently, \citet{castiglioni2020online} propose to relax the assumption by framing Bayesian persuasion into an online learning framework, focusing on the basic \emph{single-receiver} problem.\footnote{A recent work by~\citet{babichenko2021regret} relaxes the assumption in the offline setting. In that work, the goal is minimizing the sender's regret over a single iteration, and the authors provide positive results for the case in which the sender knows the ordinal preferences of the receiver over states of nature. The authors study the case of a single receiver with a binary action space, and an arbitrary (unknown) utility function.}
In their model, the sender repeatedly faces a receiver whose type during each iteration---determining her/his utility function---is unknown and adversarially selected beforehand.
In this work, we extend the model by~\citet{castiglioni2020online} to \emph{multi-receiver} settings, where the (unknown) type of each receiver is adversarially selected before each iteration of the repeated interaction.
We consider the case in which the sender has a \emph{private} communication channel towards each receiver, which is commonly studied in multi-receiver models (see, \emph{e.g.},~\citep{babichenko2016computational}). 
Dealing with multiple receivers introduces the additional challenge of correlating information disclosure across them and requires different techniques from those used in the single-receiver setting.

As customary when studying multi-receiver Bayesian persuasion problems~\citep{dughmi2017algorithmic,Xu2020tractability}, we address the case in which there are \emph{no inter-agent externalities}, where each receiver's utility does \emph{not} depend on the actions of the other receivers, but only on her/his own action and the state of nature.
Moreover, we focus on the commonly-studied setting with \emph{binary actions}~\citep{babichenko2016computational,arieli2019private}, and we analyze different scenarios depending on whether the sender's utility function is supermodular, submodular, or anonymous.
Despite its simplicity, this basic model encompasses several real-world scenarios.
For instance, think of a marketing problem in which a firm (sender) wants to persuade some potential buyers (receivers) to buy one of its products.
Each buyer has to take a binary decision as to whether to buy a unit of the product or not, while the firm's goal is to strategically disclose information about the product to the buyers, so as to maximize the number of units sold.
In this example, the sender's utility is anonymous, since it only depends on the number of buyers who decide to purchase (and not on their identities).
Moreover, submodular sender's utilities represent diminishing returns in the number of items sold, while supermodular ones encode decreasing production costs.

\subsection{Original Contributions}

Our goal is to design online algorithms for the sender that recommend a signaling scheme at each iteration of the repeated interaction, guaranteeing a sender's expected utility close to that of the best-in-hindsight signaling scheme.
In particular, we look for no-$\alpha$-regret algorithms, which collect an overall utility that is close to a fraction $\alpha$ of what can be obtained by the best-in-hindsight signaling scheme.
In this work, we assume \emph{full-information} feedback, which means that, after each iteration, the sender observes each receiver's type during that iteration.
Moreover, we are interested in no-$\alpha$-regret algorithms having a per-iteration running time polynomial in the size of the problem instance.
To this end, we assume that the number of possible types of each receiver is fixed, otherwise polynomial-time no-$\alpha$-regret algorithms cannot be obtained even in the degenerate case of only one receiver~\citep{castiglioni2020online}.

In Section~\ref{sec:hardness}, we prove a negative result: for any $0 < \alpha \leq 1$, there is no polynomial-time no-$\alpha$-regret algorithm when the sender's utility function is \emph{supermodular} or \emph{anonymous}.
Thus, in the rest of the work, we focus on the case in which the sender's utility function is \emph{submodular}, where we provide a polynomial-time no-$\left( 1 - \frac{1}{e} \right)$-regret algorithm.\footnote{Our result is tight, as there is no poly-time no-$\alpha$-regret algorithm with $\alpha > 1 - \frac{1}{e}$. Indeed, it is \NPHard~to approximate the sender's optimal utility within a factor $> 1 - \frac{1}{e} $, even in the basic (offline) multi-receiver model of~\citet{babichenko2016computational}.}

As a first step in building our algorithm, in Section~\ref{sec:learning_oracle} we introduce a general \emph{online gradient descent} (OGD) scheme to handle online learning problems with a finite number of possible loss functions.
This can be applied to our setting, as we have a sender's utility function (or, equivalently, negative loss function) for every combination of receivers' types obtained as feedback.
The OGD scheme works in a modified decision space whose dimensionality is the number of observed loss functions, and it is \emph{not} affected by the dimensionality of the original space.
This is crucial in our setting, as it avoids dealing with the set of sender's signaling schemes, whose dimensionality grows exponentially in the number of receivers.
Any OGD algorithm requires a projection oracle.
Since in our setting an exact oracle cannot be implemented in polynomial time, we build our OGD scheme so that it works having access to a suitably-defined \emph{approximate projection oracle}, which, as we show later, can be implemented in polynomial time in our model.

In Section~\ref{sec:projection_oracle}, we build a polynomial-time approximate projection oracle.
First, we formulate the projection problem as a convex linearly-constrained quadratic program, which has exponentially-many variables and polynomially-many constraints.
Next, we show how to compute in polynomial time an approximate solution to this program by applying the \emph{ellipsoid algorithm} to its dual.
%
%
Since the dual has polynomially-many variables and exponentially-many constraints, the algorithm needs access to a particular (problem-dependent) polynomial-time separation oracle.
Unfortunately, we do not have this in our setting, and, thus, our algorithm must rely on an \emph{approximate separation oracle}.
In general, running the ellipsoid method with an approximate separation oracle does \emph{not} give any guarantee on the approximation quality of the returned solution.
In order to make the ellipsoid algorithm return the desired approximate solution by only using an approximate separation oracle, we employ some \emph{ad-hoc} technical tools suggested by a non-trivial primal-dual analysis.
As a preparatory step towards our main result, at the beginning of Section~\ref{sec:projection_oracle}, we use a derivation similar to that described so far to design a polynomial-time approximation algorithm for the offline version of our multi-receiver Bayesian persuasion problem, which may be of independent interest.

In Section~\ref{sec:submodular}, we conclude the construction of the no-$\left( 1 - \frac{1}{e} \right)$-regret algorithm by showing how to implement in polynomial time an $\left( 1- \frac{1}{e} \right)$-approximate separation oracle for settings in which the sender's utility is submodular.

All the proofs omitted from the paper are in the Appendix.

\subsection{Related Works}

Most of the computational works on Bayesian persuasion study (offline) models in which the sender knowns the receiver's utility function exactly.
\citet{dughmi2016algorithmic} initiate these studies with the single receiver case, while \citet{arieli2019private} extend their work to multiple receivers without inter-agent externalities, with a  focus on private signaling.
In particular, they focus on settings with binary actions for the receivers and a binary space of states of nature. They provide a characterization of the optimal signaling scheme in the case of supermodular, anonymous submodular, and super-majority sender's utility functions.
\citet{babichenko2016computational} extend this latter work by providing tight $(1-\frac{1}{e})$-approximate signaling schemes for monotone submodular sender's utilities and showing that an optimal private signaling scheme for anonymous utility functions can be found efficiently.
\citet{dughmi2017algorithmic} generalize the previous model to settings with an arbitrary number of states of nature.
There are also some works focusing on public signaling with no inter-agent externalities, see, among others, \citep{dughmi2017algorithmic}~and~\citep{Xu2020tractability}. 
%
%
%

The only computational work on Bayesian persuasion in an online learning framework is that of~\citet{castiglioni2020online}, which, however, is restricted to the single-receiver case.
The results and techniques in~\citet{castiglioni2020online} are different from those in our paper.
In particular, they show that there are no polynomial-time no-$\alpha$-regret algorithms even in settings with a single receiver, when the number of receiver's types is arbitrary.
In contrast, we focus on settings in which the number of receivers' types is fixed.
%
%
Moreover, the main goal of \citet{castiglioni2020online} is to design a (necessarily exponential-time) no-regret algorithm in the partial-information feedback setting in which the sender only observes the actions played by the receiver (and \emph{not} her/his types).
This is accomplished by providing slightly-biased estimators of the sender's utilities for different signaling schemes.
In our work, we assume full-information feedback, and, thus, our main focus is dealing with multiple receivers.
This also introduces the additional challenge of correlating information disclosure across the receivers and working with an exponential number of possible feedbacks (tuples specifying a type for each receiver).

Our work is also related to the research line on \emph{online linear optimization with approximation oracles}.
In such setting, \citet{kakade2009playing} show how to design a no-$\alpha$-regret algorithm relying on an $\alpha$-approximate linear optimization oracle, while \citet{garber2017}~and~\citet{Hazan2018online} obtain analogous results with a better query complexity. 
The approach of these works to design approximate projection oracles is fundamentally different from ours, since they have access to a linear optimization oracle working in the learner's decision space.
On the other hand, our OGD scheme works on a modified decision space, and the approximate projection oracle can only rely on an approximate sepration oracle dealing with the original sender's decision space.


%

\section{Preliminaries}\label{sec:preliminaries}

There is a finite set $\rec \defeq \{r_i\}_{i=1}^{n}$ of $n$ receivers, and each receiver $r \in \rec$ has a type chosen from a finite set $\K_r \defeq \{k_{r, i} \}_{i=1}^{m_r}$ of $m_r$ different types (let $m \defeq \max_{r \in \rec} m_r$).
We introduce $\K \defeq \bigtimes_{r \in \rec} \K_r$ as the set of type profiles, which are tuples $\kvec \in \K$ defining a type $k_r \in \K_r$ for each receiver $r \in \rec$.\footnote{All vectors and tuples are denoted by bold symbols. For any vector (tuple) $\xvec$, the value of its $i$-th component is denoted by $x_i$.}
Each receiver $r \in \rec$ has two actions available, defined by $\A_r \defeq \{a_0,a_1\}$.
We let $\A \defeq \bigtimes_{r \in \rec} \A_r$ be the set of action profiles specifying an action for each receiver.
%
%
%
Sender and receivers' payoffs depend on a random state of nature, which is selected from a finite set $\Theta\defeq\{\theta_i\}_{i=1}^\nState$ of $d$ states.
The payoff of a receiver also depends on the action played by her/him, while it does \emph{not} depend on the actions played by the other receivers, since there are \emph{no inter-agent externalities}.
Formally, a receiver $r \in \rec$ of type $k \in \K_r$ has a utility $u^{r,k}: \A_r \times\Theta\to [0,1]$.
For the ease of notation, we let $u_\theta^{r,k} \defeq u^{r,k}(a_1, \theta)-u^{r,k}(a_0,\theta)$ be the payoff difference for a receiver $r $ of type $k $ when the state of nature is $\theta \in \Theta$.
%
%
%
The sender's utility depends on the actions played by all the receivers, and it is defined by $\usend : \A \times \Theta \to [0,1]$.
For the ease of presentation, for every state $\theta \in \Theta$, we introduce the function $f_\theta: 2^\rec \to [0,1]$ such that $f_\theta(R)$ represents the sender's utility when the state of nature is $\theta$ and all the receivers in $R \subseteq \rec$ play action $a_1$, while the others play $a_0$.

As it is customary in Bayesian persuasion, we assume that the state of nature is drawn from a common prior distribution $\muvec \in\textnormal{int}(\Delta_\Theta)$, which is explicitly known to both the sender and the receivers.\footnote{$\textnormal{int}(X)$ is the \emph{interior} of a set $X$, while $\Delta_X$ is the set of all the probability distributions over a set $X$.}
The sender can commit to a {\em signaling scheme} $\phi$, which is a randomized mapping from states of nature to signals for the receivers.
In this work, we focus on \emph{private} signaling, where each receiver has her/his own signal that is privately communicated to her/him.
Formally, there is a finite set $\sset_r$ of possible signals for each receiver $r \in \rec$.
Then, $\phi:\Theta \to \Delta_\sset$, where $\sset \defeq \bigtimes_{r \in \rec} \sset_r$ is the set of signal profiles, which are tuples $\svec \in \sset$ defining a signal $s_r \in \sset_r$ for each receiver $r \in \rec$. 
We denote with $\phi_\theta$ the probability distribution employed by $\phi$ when the state of nature is $\theta \in \Theta$, with $\phi_\theta(\svec)$ being the probability of sending a signal profile $\svec \in \sset$.
The \emph{one-shot} interaction between the sender and the receivers goes on as follows: \emph{(i)} the sender commits to a publicly known signaling scheme $\phi$; \emph{(ii)} she/he observes the realized state of nature $\theta \sim \muvec$; \emph{(iii)} she/he draws a signal profile $\svec \sim \phi_\theta$ and communicates to each receiver $r \in \rec$ signal $s_r$; and \emph{(iv)} each receiver $r \in \rec$ rationally updates her/his prior belief over $\Theta$ according to the {\em Bayes rule} and selects an action maximizing her/his expected utility.
%
%
%
We remark that, given a signaling scheme $\phi$, a receiver $r \in \rec$ of type $k \in \K_r$ observing a private signal $s \in \sset_r$ experiences an expected utility $\sum_{\theta\in\Theta}\mu_\theta \sum_{\svec \in \sset : s_r = s} \phi_\theta(\svec)  \,u^{r,k}(a, \theta)$ (up to a normalization constant) when playing action $a \in \A_r$.
Assuming the receivers' type profile is $\kvec \in \K$, the goal of the sender is to commit to an \emph{optimal} signaling scheme $\phi$, which is one maximizing her/his expected utility $f(\phi,\kvec) \defeq \sum_{\theta\in\Theta}\mu_\theta \sum_{\svec \in \sset} \phi_\theta(\svec) \,f_\theta ( R_{\svec}^{\kvec} )$, where we let $R_{\svec}^{\kvec} \subseteq \rec$ be the set of receivers who play $a_1$ after observing their private signal $s_r$ in $\svec$, under signaling scheme $\phi$.
%

\paragraph{Assumptions}
In the rest of this work, we assume that the the sender's utility is \emph{monotone non-decreasing} in the set of receivers playing $a_1$.
Formally, for each state $\theta \in \Theta$, we let $f_\theta(R) \le f_\theta(R')$ for every $R \subseteq R' \subseteq \rec$, while $f_\theta(\varnothing) = 0$ for the ease of presentation.
Moreover, we assume that the number of types $m_r$ of each receiver $r \in \rec$ is fixed; in other words, the value of $m$ cannot grow arbitrarily large.\footnote{The monotonicity assumption is w.l.o.g. for this work, since our main positive result (Theorem~\ref{thm:no_regret}) relies on it. Instead, assuming a fixed number of types is necessary, since, even in single-receiver settings, designing no-regret algorithms with running time polynomial in $m$ is intractable~\citep{castiglioni2020online}.}

\paragraph{Direct Signaling Schemes}
By well-known revelation-principle-style arguments~\citep{kamenica2011bayesian,arieli2019private}, we can restrict our attention to signaling schemes that are {direct} and {persuasive}.
In words, a signaling scheme is \emph{direct} if signals correspond to recommendations of playing actions, while it is \emph{persuasive} if the receivers do not have any incentive to deviate from the recommendations prescribed by the signals they receive.
In our setting, a direct signal sent to a receiver specifies an action recommendation for each receiver's type; thus, we let $\sset_r \defeq 2^{\K_r}$ for every $r \in \rec$.
A signal $s \in \sset_r$ for a receiver $r \in \rec$ is encoded by a subset of her/his types, namely $s \subseteq \K_r$.
Intuitively, $s$ can be interpreted as the recommendation to play action $a_1$ when the receiver has type $k \in \K_r$ such that $k \in s$, while $a_0$ otherwise.
Given a direct and persuasive signaling scheme $\phi$, for a signal profile $\svec \in \sset$ and a type profile $\kvec \in \K$, the set $R_{\svec}^{\kvec}$ appearing in the definition of the sender's expected utility $f(\phi,\kvec)$ can be formally expressed as $R_{\svec}^{\kvec} \defeq \left\{   r \in \rec \mid k_r \in s_r  \right\}$.

\paragraph{Set Functions and Matroids}
In Section~\ref{sec:submodular}, we show how to implement our approximate separation oracle by optimizing functions $f_\theta$ over suitably defined matroids (representing signals).
Next, we introduce the necessary definitions on set functions and matroids.
For the ease of presentation, we consider a generic function $f: 2^\G \to [0,1]$ for a finite set $\G$.
%
%
%
%
We say that $f$ is \emph{submodular}, respectively \emph{supermodular}, if for $I,I' \subseteq \G$: $f(I\cap I')+f(I\cup I')\le f(I)+f(I')$, respectively $f(I\cap I')+f(I\cup I')\ge  f(I)+f(I')$.
The function $f$ is \emph{anonymous} if $f(I) = f(I')$ for all $I,I' \subseteq \G: |I| = |I'|$.
A matroid $\M \defeq (\G,\Ical)$ is defined by a finite ground set $\G$ and a collection $\Ical$ of independent sets, \emph{i.e.}, subsets of $\G$ satisfying some characterizing properties (see~\citep{schrijver2003combinatorial} for a detailed formal definition).
We denote by $\B(\M)$ the set of the \emph{bases} of $\M$, which are the maximal sets in $\Ical$.

\section{Multi-Receiver Online Bayesian Persuasion}\label{sec:online_multi}

We consider a multi-receiver generalization of the online setting introduced by~\citet{castiglioni2020online}.
The sender plays a repeated game in which, at each iteration $t \in [T]$, she/he commits to a signaling scheme $\phi^t$, observes the realized state of nature $\theta^t\sim \muvec$, and privately sends signals determined by $\svec^t\sim\phi_{\theta^t}^t$ to the receivers.\footnote{Throughout the paper, the set $\{1,\dots,x\}$ is denoted by $[x]$.}
Then, each receiver (whose type is unknown to the sender) selects an action maximizing her/his expected utility given the observed signal (in the {\em one-shot} interaction at iteration $t$).

We focus on the problem of computing a sequence $\{\phi^t\}_{t \in [T]}$ of signaling schemes maximizing the sender's expected utility when the sequence of receivers' types $ \{\kvec^t\}_{t\in [T]}$, with $\kvec^t \in\K$, is adversarially selected beforehand.
After each iteration $t \in [T]$, the sender gets payoff $f(\phi^t, \kvec^t)$ and receives a {\em full-information feedback} on her/his choice at $t$, which is represented by the type profile $\kvec^t$.
Therefore, after each iteration, the sender can compute the expected utility $f(\phi,\kvec^t)$ guaranteed by any signaling scheme $\phi$ she/he could have chosen during that iteration.

We are interested in an algorithm computing $\phi^t$ at each iteration $t \in [T]$.
We measure the performance of one such algorithm using the $\alpha$-\emph{regret} $R_\alpha^T$.
%
%
Formally, for $0 < \alpha \leq 1$,
\[ 
	R_\alpha^T\defeq \alpha \max_\phi  \sum_{t \in [T]} f(\phi,\kvec^t)- \Expec\left[\sum_{t \in [T]} f(\phi^t,\kvec^t)\right],
\] 
%
%
%
where the expectation is on the randomness of the algorithm.
The classical notion of regret is obtained for $\alpha = 1$.
%
%
%
%

Ideally, we would like an algorithm that returns a sequence $\{\phi^t\}_{t \in [T]}$ with the following properties:
\begin{itemize}
	\item the $\alpha$-regret is sublinear in $T$ for some $0 < \alpha \leq 1$;
	%
	%
	%
	\item the number of computational steps it takes to compute $\phi^t$ at each iteration $t \in [T]$ is $\mathsf{poly}(T,n,d)$, that is, it is a polynomial function of the parameters $T$, $n$, and $d$
\end{itemize}
An algorithm satisfying the first property is called a \emph{no-$\alpha$-regret algorithm} (it is \emph{no-regret} if it does so for $\alpha=1$). 
In this work, we focus on the weaker notion of $\alpha$-regret since, as we discuss next, requiring no-regret is oftentimes too limiting in our setting (from a computational perspective).

\section{Hardness of Being No-$\alpha$-Regret}\label{sec:hardness}

We start with a negative result.
We show that designing no-$\alpha$-regret algorithms with polynomial per-iteration running time is an intractable problem (formally, it is impossible  unless \NP~$\subseteq$~\RP) when the sender's utility is such that functions $f_\theta$ are \emph{supermodular} or \emph{anonymous}.
%
%
This hardness result is deeply connected with the intractability of the offline version of our multi-receiver Bayesian persuasion problem that we formally define in the following Section~\ref{sec:hardness_offline_def}.
Then, Section~\ref{sec:hardness_theorems} collects all the hardness results.

\subsection{Offline Multi-Receiver Bayesian Persuasion}\label{sec:hardness_offline_def}

We consider an offline setting where the receivers' type profile $\kvec \in \K$ is drawn from a known probability distribution (rather then being selected adversarially at each iteration).
Given a subset of possible type profiles $K \subseteq \K$ and a distribution $\lambdavec \in \textnormal{int}(\Delta_{K})$, we call \textsf{BAYESIAN-OPT-SIGNAL} the problem of computing a signaling scheme that maximizes the sender's expected utility.
This can be achieved by solving the following LP of exponential size.\footnote{Constraints~\eqref{LP1:cons_pers} encode persuasiveness for the signals recommending to play $a_1$. The analogous constraints for $a_0$ can be omitted. Indeed, by assuming that each $f_\theta$ is non-decreasing in the set of receivers who play $a_1$, any signaling scheme in which the sender recommends $a_0$ when the state is $\theta$ and the receiver prefers $a_1$ over $a_0$ can be improved by recommending $a_1$ instead.}
\begin{subequations}\label{LP1}
	\begin{align}  
	\max_{\phi} & \quad \sum_{\kvec \in \K} \lambda_{\kvec} \sum_{\theta\in\Theta} \mu_\theta \sum_{\svec \in \sset}\phi_\theta(\svec) f_\theta ( R_{\svec}^{\kvec} ) \\
	\textnormal{s.t. } & \sum_{\theta\in\Theta}\mu_\theta \sum_{\svec \in \sset: s_r = s} \phi_\theta(\svec) u^{r,k}_\theta \geq 0 \nonumber \\
	&\hspace{1.6cm} \forall r \in \rec, \forall s \in \sset_r, \forall k \in \K_r : k \in s \label{LP1:cons_pers}\\
	&\sum_{\svec \in \sset}\phi_{\theta}(\svec)= 1 \hspace{3.35cm} \forall \theta \in \Theta\label{LP1:constr_simplex}\\
	&\phi_{\theta}(\svec)\ge 0 \hspace{2.8cm} \forall \theta \in \theta, \forall \svec \in \sset.
	\end{align} 
\end{subequations}
%
%
%

\subsection{Hardness Results}\label{sec:hardness_theorems}

First, we study the computational complexity of finding an approximate solution to \textsf{BAYESIAN-OPT-SIGNAL}.
In particular, given $0 < \alpha \leq 1$, we look for an $\alpha$-approximate solution in the multiplicative sense, \emph{i.e.}, a signaling scheme providing at least a fraction $\alpha$ of the sender's optimal expected utility (the optimal value of LP~\eqref{LP1}).
Theorem~\ref{thm:hardness} provides our main hardness result, which is based on a reduction from the \emph{promise-version} of \textsf{LABEL-COVER} (see Appendix~\ref{sec:app_hardness} for its definition and the proof of the theorem).

\begin{restatable}{theorem}{thmhardness}\label{thm:hardness}
	For every $0 < \alpha \leq 1$, it is \NPHard~to compute an $\alpha$-approximate solution to \textnormal{\textsf{BAYESIAN-OPT-SIGNAL}}, even when the sender's utility is such that, for every $\theta \in \Theta$, $f_\theta(R)= 1$ iff $|R| \geq 2$, while $f_\theta(R)= 0$ otherwise.
	%
	%
\end{restatable}

Notice that Theorem~\ref{thm:hardness} holds for problem instances in which functions $f_\theta$ are anonymous.
Moreover, the reduction can be easily modified so that functions $f_\theta$ are supermodular and satisfy $f_\theta(R)=\max \{0,|R|-1 \}$ for $R \subseteq \rec$.
Thus:
%
%

\begin{corollary}
	For $0 < \alpha \leq 1$, it is \NPHard~to compute an $\alpha$-approximate solution to \textnormal{\textsf{BAYESIAN-OPT-SIGNAL}}, even when the sender's utility is such that functions $f_\theta$ are \emph{supermodular} or \emph{anonymous} for every $\theta \in \Theta$.
	%
\end{corollary}

By using arguments similar to those employed in the proof of Theorem 6.2 by~\citet{roughgarden2019minimizing}, the hardness of computing an $\alpha$-approximate solution to the offline problem can be extended to designing no-$\alpha$-regret algorithms in the online setting.
Then:
%

\begin{theorem}\label{thm:hardness_learning}
	For every $0 < \alpha \leq 1$, there is no polynomial-time no-$\alpha$-regret algorithm for the multi-receiver online Bayesian persuasion problem, unless \NP~$\subseteq$~\RP, even when functions $f_\theta$ are \emph{supermodular} or \emph{anonymous} for all $\theta \in \Theta$.
	%
\end{theorem}

In the rest of the work, we show how to design a polynomial-time no-$(1 - \frac{1}{e})$-regret algorithm for the case in which the sender's utility is such that functions $f_\theta$ are submodular.

\section{An Online Gradient Descent Scheme with Approximate Projection Oracles} \label{sec:learning_oracle}

As a first step in building our polynomial-time algorithm, we introduce our OGD scheme with an \emph{approximate projection oracle}.
Intuitively, it works by transforming the multi-receiver online Bayesian persuasion setting into an equivalent online learning problem whose decision space does not need to explicitly deal with signaling schemes (thus avoiding the burden of having an exponential number of possible signal profiles).
The OGD algorithm is then applied on this new domain.
In our setting, we do \emph{not} have access to a polynomial-time (exact) projection oracle, and, thus, we design and analyze the algorithm assuming access to an {approximate} one only.
As we show later in Sections~\ref{sec:projection_oracle}~and~\ref{sec:submodular}, such approximate projection oracle can be implemented in polynomial time when the functions $f_\theta$ are submodular.

Let us recall that the OGD scheme that we describe in this section is general and applies to any online learning problem with a finite number of possible loss functions.

%
%
%

%
%

\subsection{A General Approach}\label{sec:learning_oracle_framework}

Consider an online learning problem in which the learner takes a decision $y^t \in \Y$ at each iteration $t \in [T]$.
Then, the learner observes a feedback $e^t \in \E$, where $\E$ is a finite set of $p$ possible feedbacks.
The reward (or negative loss) of a decision $y \in \Y$ given feedback $e \in \E$ is defined by $u(y,e)$ for a given function $u : \Y \times \E \to [0,1]$.
Thus, the learner is awarded $u(y^t,e^t)$ for decision $y^t$ at iteration $t$, while she/he would have achieved $u(y,e^t)$ for any other choice $y \in \Y$.

We transform this general online learning problem to a new one in which the learner's decision set is $\X\subseteq [0,1]^p$ with:
\begin{equation}\label{eq:def_x}
\X \defeq \bigcup_{y \in \Y} \Big\{  \xvec \in [0,1]^p \mid x_{e} \leq u(y,e) \quad \forall e \in \E \Big\}.
\end{equation}
Intuitively, the set $\X$ contains all the vectors whose components $x_e$ (one for each feedback $e \in \E$) are the learner's rewards $u(y,e)$ for some decision $y \in \Y$ in the original problem.
Moreover, the inequality ``$\leq$'' in the definition of $\X$ also includes all the reward vectors that are dominated by those corresponding to some decision in $\Y$.
At each iteration $t \in [T]$, the learner takes a decision $\xvec^t \in \X$ and observes a feedback $e^t \in \E$.
The reward of decision $\xvec \in \X$ at iteration $t$ is the $e^t$-th component of $\xvec$, namely $x_{e^t}$.
It is sometimes useful to write it as $ \ind_{e^t}^\top \xvec$, where $\ind_{e^t} \in \{0,1\}^{p}$ is a vector whose $e^t$-th component is $1$, while all the others are $0$.
Thus, the learner's reward at iteration $t$ is $x^t_{e^t}$.
%
%
Notice that the size of the decision set $\X$ of the new online learning setting does \emph{not} depend on the dimensionality of the original decision set $\Y$ (which, in our setting, would be exponential), but only on the number of feedbacks $p$.

If $\Y$ and $u$ are such that $\X$ is compact and convex, then we can minimize the $\alpha$-regret $R_\alpha^T$ in the original problem by doing that in the new setting.
%
%
Let us introduce the set $\alpha \X \defeq \{ \alpha\xvec \mid  \xvec \in \X\}$ for any $0 <\alpha \leq 1$.
Given a sequence of feedbacks $\{e^t\}_{t \in [T]}$ and a sequence of decisions $\{\xvec^t\}_{t \in [T]}$, with $e^t \in \E$ and $\xvec^t \in \X$, we have that:
\begin{align*}
	R^T_\alpha & \defeq  \max_{\xvec \in \alpha\X} \sum_{t \in [T]} \ind_{e^t}^\top \Big( \xvec - \xvec^t \Big)  \\
	& \geq  \alpha \max_{y \in \Y} \sum_{t \in [T]}  u(y,e^t) - \sum_{t \in [T]} u(y^t,e^t), 
\end{align*}
where $\{y^t\}_{t \in [T]}$ is a sequence of decisions $y^t \in \Y$ for the original problem such that $x_e^t \leq u(y^t,e)$ for $e \in \E$.
%
%
%
%
%
%
%
%

We assume to have access to an approximate projection oracle for $\alpha \X$, which we define in the following.
%
%
%
By letting $E \subseteq \E$ be a subset of feedbacks, we define $\tau_E: \X \to [0,1]^p$ as the function mapping any vector $\xvec \in \X$ to another one that is equal to $\xvec$ in all the components corresponding to feedbacks $e \in E$, while it is $0$ everywhere else.
Moreover, we let $\X_E \defeq \{ \tau_E(\xvec) \mid \xvec \in \X \}$ be the image of $\X$ through $\tau_E$, while $\alpha \X_E \defeq \{ \alpha\xvec \mid  \xvec \in \X_E\}$ for $0 < \alpha \leq 1$.
\begin{definition}[Approximate projection oracle]\label{def:apx_proj_oracle}
	%
	Consider a subset of feedbacks $E \subseteq \E$, a vector $\yvec \in [0,2]^p$ such that $y_e = 0$ for all $e \notin E$, and an approximation error $\epsilon \in \mathbb{R}_+$.
	Then, for any $0 < \alpha \leq 1$, an \emph{approximate projection oracle} $\varphi_\alpha (E,\yvec,\epsilon)$ is an algorithm returning a vector $\xvec \in \X_E$ and a decision $y \in \Y$ with $x_e \leq u(y,e)$ for all $e \in \E$, such that:
	\[
		\vert \vert \xvec'-\xvec \vert \vert^2 \le \vert \vert \xvec'-\yvec \vert \vert^2+\epsilon \quad\quad  \forall \xvec' \in \alpha \X_E.
	\]
\end{definition}
%
%
%
%
Intuitively, $\varphi_{\alpha}$ returns a vector $\xvec \in \X_E$ that is an {approximate} projection of $\yvec$ onto the subspace $\alpha \X_E$.
The vector $\xvec$ can be outside of $\alpha \X_E$.
However, it is ``better'' than a projection onto $\alpha \X_E$, since, ignoring the $\epsilon$ error, $\xvec$ is closer than $\yvec$ to any vector in $\alpha \X_E$.  
Moreover, $\varphi_{\alpha}$ also gives a decision $y \in \Y$ that corresponds to the returned vector $\xvec$.
Notice that, if $\alpha=1$ and $\epsilon=0$, this is equivalent to find an exact projection onto the subspace $\X_E$.
%
%
%
%
%
%
%
%
%
%

\subsection{A Particular Setting: Multi-Receiver Online Bayesian Persuasion}
Our setting can be easily cast into the general learning framework described so far.
The possible feedbacks are type profiles, namely $\E \defeq \K$, while the receivers' type profile $\kvec^t \in \K$ is the feedback observed at iteration $t \in [T]$, namely $e^t \defeq \kvec^t$.
Notice that the number of possible feedbacks is $p$ is $m^n$, which is exponential in the number of receivers.
The decision set of the learner (sender) $\Y$ is the set of all the possible signaling schemes $\phi$, with $y^t \defeq \phi^t$ being the one chosen at iteration $t$.
The rewards observed by the sender are the utilities $f(\phi,\kvec)$; formally, for every signaling scheme $\phi$ and type profile $\kvec \in \K$, which define a pair $y \in \Y$ and $e \in \E$ using the generic notation, we let $u(y,e) \defeq f(\phi,\kvec)$.
Then, the new decision set $\X \subseteq [0,1]^{|\K|}$ is defined as in Equation~\eqref{eq:def_x}.
Notice that $\X$ is a compact and convex set, since it can be defined by a set of linear inequalities.
In the following, we overload the notation and, for any subset $K \subseteq \K$ of types profiles, we let $\X_K \defeq \X_E$ for $E \subseteq \E : E = K$. 

\subsection{OGD with Approximate Projection Oracle}\label{sec:learning_oracle_algorithm}

Algorithm~\ref{alg:OGD} is an OGD scheme that operates in the $\X$ domain by having access to an approximate projection oracle $\varphi_\alpha$ (we call the algorithm \textsc{OGD}-\textsc{APO}).

The procedure in Algorithm~\ref{alg:OGD} keeps track of the set $E^t \subseteq \E$ of different feedbacks observed up to each iteration $t \in [T]$.
Moreover, it works on the subspace $\X_{E^t}$, whose vectors are zero in all the components corresponding to feedbacks $e \notin E^t$.
Since it is the case that $|E^t| \leq t$, the procedure in Algorithm~\ref{alg:OGD} attains a per-iteration running time that is independent of the number of possible feedbacks $p$.

%
%

\begin{algorithm}[H]
	\caption{\textsc{OGD}-\textsc{APO}}
	\label{alg:OGD}
	\begin{algorithmic}
		\STATE \begin{tabular}{ll}
			\hspace{-2.2mm}{\bfseries Input:} & $\bullet$ approximate projection oracle $\varphi_\alpha$ \\
			& $\bullet$ learning rate $\eta \in (0,1]$ \\
			& $\bullet$ approximation error $\epsilon \in [0,1]$ 
		\end{tabular}
		\STATE
		\hrule
		\STATE Initialize $y^1 \in \Y$, $E^0 \gets \varnothing$, and $\xvec^1 \gets \boldsymbol{0} \in \X_{E^1}$ 
		\FOR{$t = 1,\dots, T$}
		\STATE Take decision $y^t$
		\STATE Observe feedback $e^t \in \E$ and reward $u(y^t,e^t) = x^t_{e^t}$
		\STATE $E^{t} \gets E^{t-1}\cup \{ e^t\}$
		\STATE $\yvec^{t+1} \gets \xvec^t + \eta 1_{e^t}$
		\STATE $\left( \xvec^{t+1},y^{t+1} \right) \gets \varphi_{\alpha} \left( E^t,\yvec^{t+1},\epsilon \right)$ 
		\ENDFOR 
	\end{algorithmic}
\end{algorithm}

Next, we bound the $\alpha$-regret incurred by Algorithm~\ref{alg:OGD}.

\begin{restatable}{theorem}{thmregretogd} \label{thm:regretogd}
	Given an oracle $\varphi_{\alpha}$ (as in Definition~\ref{def:apx_proj_oracle}) for some $0 < \alpha \leq 1$, a learning rate $\eta \in (0,1]$, and an approximation error $\epsilon\in[0,1]$, Algorithm~\ref{alg:OGD} has $\alpha$-regret
	\[
		R^T_{\alpha} \leq \frac{|E^T|}{2\eta }+\frac{\eta T}{2} + \frac{\epsilon T}{2 \eta}, 
	\]
	with a per-iteration running time $\mathsf{poly}(t)$.
\end{restatable}

By setting $\eta = \frac{1}{\sqrt{T}}$, $\epsilon = \frac{1}{T}$, we get $R^T_{\alpha} \leq \sqrt  T \left( 1+ \frac{|E^T|}{2} \right)$.

Notice that the bound only depends on the number of observed feedbacks $|E^T|$, while it is independent of the overall number of possible feedbacks $p$.
This is crucial for the multi-receiver online Bayesian persuasion case, where $p$ is exponential in the the number of receivers $n$.
On the other hand, as $T$ goes to infinity, we have $|E^T| \leq p$, so that the regret bound is sublinear in $T$.


\section{Constructing a Poly-Time Approximate Projection Oracle}\label{sec:projection_oracle}

The crux of the OGD-APO algorithm (Algorithm~\ref{alg:OGD}) is being able to perform the approximate projection step.
In this section, we show that, in the multi-receiver Bayesian persuasion setting, the approximate projection oracle $\varphi_{\alpha}$ required by \textsc{OGD}-\textsc{APO} can be implemented in polynomial time by an appropriately-engineered ellipsoid algorithm. 
This calls for an {\em approximate separation oracle} $\Ocal_\alpha$ (see Definition~\ref{def:oracle}).
%

We proceed as follows.
In Section~\ref{sec:offline}, we define an appropriate notion of approximate separation oracle, and show how to find, in polynomial time, an $\alpha$-approximate solution to the offline problem \textsf{BAYESIAN-OPT-SIGNAL}. 
This is a preparatory step towards the understanding of our main result in this section, and it may be of independent interest.
Then, in Section~\ref{sec:sep_to_proj}, we exploit some of the techniques introduced for the offline setting in order to build $\varphi_{\alpha}$ starting from an approximate separation oracle $\Ocal_\alpha$.

\subsection{Warming Up: The Offline Setting}\label{sec:offline}

An approximate separation oracle $\Ocal_{\alpha}$ finds a signal profile $\svec \in \Scal$  that approximately maximizes a weighted sum of the $f_\theta$ functions,  plus a weight for each receiver which depends on the signal $s_r$ sent to that receiver.
Formally:

\begin{definition}[Approximate separation oracle] \label{def:oracle}	
	Consider a state $\theta \in \Theta$, a subset $K \subseteq \K$, a vector $\lambdavec \in \mathbb{R}_+^{|K|}$, weights $\wvec =(w_{r,s})_{r \in \rec, s \in \Scal_r}$ with $w_{r,s} \in \mathbb{R}$ and $w_{r,\varnothing}=0$ for all $r \in \rec$, and an approximation error $\epsilon\in\mathbb{R}_+$.
	Then, for any $0<\alpha\le 1$, an approximation oracle $\Ocal_{\alpha}(\theta,K,\lambdavec,\wvec,\epsilon)$ is an algorithm returning an $\svec\in\Scal$ such that:
	\begin{multline}\label{eq:approximation_oracle}
	\sum_{\kvec \in K} \lambda_{\kvec}  f_\theta(R^{\kvec}_{\svec}) + \sum_{r\in\rec} w_{r,s_r}  \\
	\ge  \max_{\svec^\star\in\Scal} \left\{\alpha \sum_{\kvec \in K} \lambda_{\kvec}  f_\theta(R^{\kvec}_{\svec^\star}) + \sum_{r\in\rec} w_{r,s^\star_r}\right\}-\epsilon,
	\end{multline}
	in time $\mathsf{poly}\left( n,|K|,\max_{r,s} |w_{r,s}|,\max_{\kvec} \lambda_{\kvec} ,\frac{1}{\epsilon}\right)$.
\end{definition}

As a preliminary result, we show how to use an oracle $\Ocal_{\alpha}$ to find in polynomial time an $\alpha$-approximate solution to \textsf{BAYESIAN-OPT-SIGNAL} (see Section~\ref{sec:hardness}).
This problem is interesting in its own right, and allows us to develop a line of reasoning that will be essential to prove Theorem~\ref{thm:proj_construction}.

\begin{restatable}{theorem}{thmOffline}\label{th:bayesian}
	Given $\epsilon\in\mathbb{R}_+$ and an approximate separation oracle $\Ocal_{\alpha}$, with $0<\alpha\le 1$, there exists a polynomial-time approximation algorithm for \emph{\textsf{BAYESIAN-OPT-SIGNAL}} returning a signaling scheme with sender's utility at least $\alpha \opt-\epsilon$, where $\opt$ is the value of an optimal signaling scheme.
	Moreover, the algorithm works in time $\mathsf{poly}(\frac{1}{\epsilon})$.
\end{restatable}	
\begin{proof}[Proof Overview]
The dual of LP~\eqref{LP1} has a polynomial number of variables and an exponential number of constraints, and a natural way to prove polynomial-time solvability would be via the {ellipsoid method} (see, {\em e.g.},~\citep{khachiyan1980polynomial,grotschel1981ellipsoid}). 
However, in our setting, we can only rely on an approximate separation oracle, which renders the traditional ellipsoid method unsuitable for our problem.
We show that it is possible to exploit a binary search scheme on the dual problem to find a value $\gamma^\star\in[0,1]$ such that the dual problem with objective $\gamma^\star$ is feasible, while the dual with objective $\gamma^\star-\beta$, $\beta\ge 0$, is infeasible. That algorithm runs in $\log(\beta)$ steps. At each iteration of the algorithm, we solve a feasibility problem through the ellipsoid method equipped with an appropriate approximate separation oracle which we design. In order to build a poly-time separation oracle we have to carefully manage all the settings in which $\Ocal_\alpha$ would not run in polynomial time, according to~Definition \ref{def:oracle}. Specifically, we need to properly manage large values of the weights $\wvec$, since $\Ocal_\alpha$ is polynomial in $\max_{r,s} |w_{r,s}|$.
Once we do that, the approximate separation oracle is guaranteed to find a violated constraint, or to certify that all constraints are approximately satisfied.
Finally, we show that the approximately feasible solution computed via bisection allows one to recover an approximate solution to the original problem
\end{proof}

\subsection{From an Approximate Separation Oracle to an Approximate Projection Oracle}\label{sec:sep_to_proj}

Now, we show how to design a polynomial-time approximate projection oracle $\varphi_{\alpha}$ using an approximate separation oracle $\Ocal_{\alpha}$.
The proof employs a convex linearly-constrained quadratic program that computes the optimal projection on $\X$, the ellipsoid method, and a careful primal-dual analysis.

\begin{restatable}{theorem}{thmProjConstruction}\label{thm:proj_construction}
	Given a subset $K \subseteq \K$, a vector $\yvec\in [0,2]^{|\K|}$ such that $y_{\kvec} = 0$ for all $\kvec \notin K$, and an approximation error $\epsilon \in \mathbb{R}_+$, for any $0<\alpha\le 1$, the approximate projection oracle $\varphi_{\alpha}(K, \yvec,\epsilon)$ can be computed in polynomial time by querying the approximate separation oracle $\Ocal_\alpha$.
\end{restatable}

\begin{proof}[Proof Overview]
We start by defining a convex minimization problem, which we denote by \circled{P}, for computing the projection of $\yvec$ on $\X_K$. Then, we work on the dual of \circled{P}, which we suitably simplify by reasoning over the KKT conditions of the problem.
As in the proof of Theorem~\ref{th:bayesian}, we proceed by repeatedly applying the ellipsoid method on a feasibility problem obtained from the dual, decreasing the required objective $\gamma^\star$ by a small additive factor $\beta$.
The ellipsoid method is equipped with the approximate separation oracle that employs the oracle in~Definition \ref{def:oracle} and carefully manages the cases in which $\Ocal_\alpha$ would not run in polynomial time. 
In this case, the problem is complicated by the fact that we have to determine an approximate projection over $\alpha\X_K$, rather than an approximate solution to \circled{P}. 
We found two dual problems such that one dual problem with objective $\gamma^\star$ is feasible, while the second one with objective $\gamma^\star+\beta$ is infeasible.
From these problems, we define a new convex optimization problem that is a modified version of \circled{P} and has value at least $\gamma^\star$. Then, we show that a solution to this problem is close to a projection on a set which includes $\alpha\X_K$.
Finally, we restrict  \circled{P} to the primal variables corresponding to the set of (polynomially-many) violated dual constraints determined during the last application of the ellipsoid method that returns unfeasible, \emph{i.e.,} where the ellipsoid method for feasibility problem is run with objective $\gamma^*+\beta$. We conclude the proof by showing that a solution to this restricted problem is precisely an approximate projection on a superset of $\alpha\X_K$.
\end{proof}

\section{A Poly-Time No-$\alpha$-Regret Algorithm for Submodular Sender's Utilities}\label{sec:submodular}

In this section, we conclude the construction of our polynomial-time no-$(1-\frac{1}{e})$-regret algorithm for settings in which sender's utilities are submodular.
The last component that we need to design is an approximate separation oracle $\Ocal_\alpha$ (see Definition~\ref{def:oracle}) running in polynomial time.
Next, we show how to obtain this by exploiting the fact that functions $f_\theta$ are submodular in the set of receivers playing action $a_1$.

First, we establish a relation between direct signals $\sset$ and matroids.
We define a matroid $\M_\sset \defeq (\G_\sset, \Ical_\sset)$ such that:
\begin{itemize}
	\item the ground set is $\G_\sset \defeq \{ (r, s) \mid r\in \rec, s \in \sset_r \}$;
	\item a subset $I \subseteq \G_\sset$ belongs to $\Ical_\sset$ if and only if $I $ contains \emph{at most one} pair for each receiver $r \in \rec$.
\end{itemize}
The elements of the ground set $\G_\sset$ represent receiver, signal pairs.
However, sets $I \in \Ical_\sset$ do \emph{not} characterize signal profiles, as they may {not} define a signal for each receiver.
Indeed, direct signal profiles are captured by the basis set $\B(\M_\sset)$ of the matroid $\M_\sset$.
Let us recall that $\B(\M_\sset)$ contains all the maximal sets in $\Ical_\sset$, and, thus, a subset $I \subseteq \Ical_\sset$ belongs to $\B(\M_\sset)$ if and only if $I$ contains \emph{exactly one} pair for each receiver $r \in \rec$.
Intuitively, a basis $I \in \B(\M_\sset)$ defines a direct signal profile $\svec \in \sset$ in which, for each receiver $r \in \rec$, all the receiver's types in $s \in \sset_r$ such that $(r, s) \in I$ are recommended to play action $a_1$, while the others are told to play $a_0$.

The following Theorem~\ref{thm:oracle_submodular} provides a polynomial-time approximation oracle $\Ocal_{1-\frac{1}{e}}$ for instances in which $f_\theta$ is submodular for each state of nature $\theta \in \Theta$.
The core idea of its proof is that $\sum_{\kvec \in K} \lambda_{\kvec}  f_\theta(R^{\kvec}_{\svec})$ (see Equation~\eqref{eq:approximation_oracle}) can be seen as a submodular function defined for the ground set $\G_\sset$ and optimizing over direct signal profiles $\svec \in \sset$ is equivalent to doing that over the bases $\B(\M_\sset)$ of the matroid $\M_\sset$.
Then, the result is readily proved by exploiting some results concerning the optimization over matroids.\footnote{The separation oracle provided in Theorem~\ref{thm:oracle_submodular} guarantees the desired approximation factor with arbitrary high probability. It is easy to see that, since the algorithm fails with arbitrary small probability, this does not modify our regret bound except for an (arbitrary small) negligible term.} 

%
%

\begin{restatable}{theorem}{thmsubmodular}\label{thm:oracle_submodular}
	If the sender's utility is such that function $f_\theta$ is \emph{submodular} for each $\theta \in \Theta$, then there exists a polynomial-time separation oracle $\Ocal_{1-\frac{1}{e}}$.
\end{restatable}

In conclusion, by letting $\K^T \subseteq \K$ be the set of receivers' type profiles observed by the sender up to iteration $T$, the following Theorem~\ref{thm:no_regret} provides our polynomial-time no-$(1-\frac{1}{e})$-regret algorithm working with submodular sender's utilities.

\begin{restatable}{theorem}{thmregretsubmodular}\label{thm:no_regret}
	If the sender's utility is such that function $f_\theta$ is \emph{submodular} for each $\theta \in \Theta$, then there exists a no-$(1-\frac{1}{e})$-regret algorithm having $(1-\frac{1}{e})$-regret
	\[
	R^T_{1-\frac{1}{e}} \leq O \left( \sqrt{T} \, |\K^T| \right),
	\]
	with a per-iteration running time $\mathsf{poly}(T,n,d)$.
\end{restatable}

\begin{proof}
	We can run Algorithm~\ref{alg:OGD} on an instance of our multi-receiver online Bayesian persuasion problem.
	By Theorem~\ref{thm:regretogd}, if we set $\eta = \frac{1}{\sqrt{T}}$, $\epsilon = \frac{1}{T}$, and $\alpha=1-\frac{1}{e}$, we get the desired regret bound (notice that the set of observed feedbacks is $E^t = \K^t$ in our setting).
	Algorithm~\ref{alg:OGD} employs an approximate projection oracle $\varphi_{1-\frac{1}{e}}$ that we can implement in polynomial time by using the algorithm provided in Theorem~\ref{thm:proj_construction}.
	This requires access to a polynomial-time approximate separation oracle $\mathcal O_{1-\frac{1}{e}}$, which can be implemented by using Theorem~\ref{thm:oracle_submodular}, under the assumption that the sender's utility is such that functions $f_\theta$ are submodular.	
\end{proof}

Notice that the regret bound only depends on the number $|\K^T|$ of receivers' type profiles observed up to iteration $T$, while it is independent of the overall number of possible type profiles $|\K| = m^n$, which is exponential in the number of receivers.
Thus, the $(1-\frac{1}{e})$-regret is polynomial in the size of the problem instance provided that the type profiles received as feedbacks by the sender are polynomially many (though the sender does not have to know which are these type profiles in advance).
This is reasonable in many practical applications, where not all the type profiles can occur, since, \emph{e.g.}, receivers' types are highly correlated.
On the other hand, let us remark that, as $T$ goes to infinity, we have $|\K^T| \leq m^n$, so that the regret is sublinear in $T$.

%

\section*{Acknowledgments}
This work has been partially supported by the Italian MIUR PRIN 2017 Project ALGADIMAR ``Algorithms, Games, and Digital Market''.

\bibliography{biblio}
\bibliographystyle{icml2021}

\clearpage

\appendix
\onecolumn

\section{Proofs Omitted from Section~\ref{sec:hardness}}\label{sec:app_hardness}

In this section, we provide the complete proof of the hardness result in Theorem~\ref{thm:hardness}.
This is based on a reduction from the promise-problem version of \textsf{LABEL-COVER}, which we define next.

The following is the formal definition of an instance of the \textsf{LABEL-COVER} problem.

\begin{definition}[\textsf{LABEL-COVER} instance]
	An instance of \textnormal{\textsf{LABEL-COVER}} consists of a tuple $( G, \Sigma, \Pi)$, where:
	\begin{itemize}
		\item $G \coloneqq (U,V,E)$ is a \emph{bipartite graph} defined by two disjoint sets of nodes $U$ and $V$, connected by the edges in $E \subseteq U \times V$, which are such that all the nodes in $U$ have the same degree;
		\item $\Sigma$ is a finite set of \emph{labels}; and
		\item $\Pi \coloneqq \left\{ \Pi_e : \Sigma \to \Sigma \mid e \in E \right\}$ is a finite set of \emph{edge constraints}.
	\end{itemize}
\end{definition}

\begin{definition}[Labeling]
	Given an instance $( G, \Sigma, \Pi)$ of \textnormal{\textsf{LABEL-COVER}}, a \emph{labeling} of the graph $G$ is a mapping $\pi: U \cup V \to \Sigma$ that assigns a label to each vertex of $G$ such that all the edge constraints are satisfied.
	Formally, a labeling $\pi$ satisfies the constraint for an edge $e =(u,v) \in E$ if $ \pi(v) = \Pi_e(\pi(u))$.
\end{definition}

The classical \textsf{LABEL-COVER} problem is the search problem of finding a valid labeling for a \textsf{LABEL-COVER} instance given as input.
In the following, we consider a different version of the problem, which is the \emph{promise problem} associated with \textsf{LABEL-COVER} instances, defined as follows.

\begin{definition}[\textsf{GAP-LABEL-COVER}$_{c,b}$]
	For any pair of numbers $0 < b < c < 1$, we define \textnormal{\textsf{GAP-LABEL-COVER}}$_{c,b}$ as the following promise problem.
	\begin{itemize}
		\item \textnormal{\texttt{Input:}} An instance $(G,\Sigma, \Pi)$ of \textnormal{\textsf{LABEL-COVER}} such that either one of the following is true:
		\begin{itemize}
			\item there exists a labeling $\pi: U \cup V \to \Sigma$ that satisfies at least a fraction $c$ of the edge constraints in $\Pi$;
			\item any labeling $\pi: U \cup V \to \Sigma$ satisfies less than a fraction $b$ of the edge constraints in $\Pi$.
		\end{itemize}
		\item \textnormal{\texttt{Output:}} Determine which of the above two cases hold.
	\end{itemize}
\end{definition}

In order to prove Theorem~\ref{thm:hardness}, we make use of the following result due to~\citet{raz1998parallel}~and~\citet{arora1998proof}.

\begin{theorem}[\citet{raz1998parallel,arora1998proof}]\label{thm:hard_gap_label}
	For any $\epsilon > 0$, there exists a constant $k_\epsilon \in \mathbb{N}$ that depends on $\epsilon$ such that the promise problem \textnormal{\textsf{GAP-LABEL-COVER}}$_{1,\epsilon}$ restricted to inputs $(G, \Sigma, \Pi)$ with $|\Sigma| = k_\epsilon$ is \textnormal{\textsf{NP}}-hard.
\end{theorem}

Next, we provide the complete proof of Theorem~\ref{thm:hardness}.

\thmhardness*

\begin{proof}
	We provide a reduction from \textnormal{\textsf{GAP-LABEL-COVER}}$_{1,\epsilon}$.
	Our reduction maps an instance $( G, \Sigma, \Pi)$ of \textsf{LABEL-COVER} to an instance of \textsf{BAYESIAN-OPT-SIGNAL} with the following properties:
	\begin{itemize}
		\item (\emph{completeness}) if the \textsf{LABEL-COVER} instance admits a labeling satisfying all the edge constraints (recall $c=1$), then the \textsf{BAYESIAN-OPT-SIGNAL} instance has a signaling scheme with sender's expected utility $\geq \left( 1-\frac{\epsilon}{|\Sigma|} \right) \frac{1}{|\Sigma|}\ge \frac{1}{2|\Sigma|} $;
		\item (\emph{soundness}) if the \textsf{LABEL-COVER} instance is such that any labeling satisfies at most a fraction $\epsilon$ of the edge constraints, then an optimal signaling scheme in the \textsf{BAYESIAN-OPT-SIGNAL} instance has sender's expected utility at most $\frac{2\epsilon}{|\Sigma|}$.
	\end{itemize}
	By Theorem~\ref{thm:hard_gap_label}, for any $\epsilon > 0$ there exists a constant $k_\epsilon \in \mathbb{N}$ that depends on $\epsilon$ such that \textsf{GAP-LABEL-COVER}$_{1,\epsilon}$ restricted to inputs $(G, \Sigma, \Pi)$ with $|\Sigma| = k_\epsilon$ is \textnormal{\textsf{NP}}-hard.
	Given $0 < \alpha \leq 1$, by setting $\epsilon = \frac{\alpha}{4}$ and noticing that $\frac{2\epsilon/|\Sigma|}{1/2|\Sigma|}=4\epsilon=\alpha$, we can conclude that it is \NPHard~to compute an $\alpha$-approximate solution to \textsf{BAYESIAN-OPT-SIGNAL}.

	\paragraph{Construction}
	Given an instance $( G, \Sigma, \Pi)$ of \textnormal{\textsf{LABEL-COVER}} defined over a bipartite graph $G \defeq (U,V,E)$, we build an instance of \textsf{BAYESIAN-OPT-SIGNAL} as follows.
	\begin{itemize}
		\item For each label $\sigma \in \Sigma$, there is a corresponding state of nature $\theta_\sigma \in \Theta$. Moreover, there is an additional state $\theta_0 \in \Theta$. Thus, the total number of possible states is $d = |\Sigma|+1$.
		\item The prior distribution is $\muvec \in \textnormal{int}(\Delta_\Theta)$ such that $\mu_{\theta_\sigma} = \frac{\epsilon}{|\Sigma|^2}$ for every $\theta_\sigma \in \Theta$ and $\mu_{\theta_0} = 1-\frac{\epsilon}{|\Sigma|}$.
		\item For every vertex $v \in U \cup V$ of the graph $G$, there is a receiver $r_v \in \rec$. Thus, $n = |U \cup V|$.
		\item Each receiver $r_v \in \rec$ has $m_{r_v} = |\Sigma|+1$ possible types. The set of types of receiver $r_v$ is $\K_{r_v} = \{ k_{\sigma} \mid \sigma \in \Sigma \} \cup \{ k_0 \}$.
		\item A receiver $r_v \in \rec$ of type $k_\sigma \in \K_{r_v}$ has utility such that $u^{r_v,k_\sigma}_{\theta_\sigma}=\frac{1}{2}$ and $u^{r_v,k_\sigma}_{\theta_{\sigma'}}=-1$ for all $\theta_{\sigma'}  \in \Theta :\theta_{\sigma'} \neq \theta_\sigma$, while $u^{r_v,k_\sigma}_{ \theta_0}=-\frac{\epsilon}{2 |\Sigma|^2}$. Moreover, a receiver $r_v \in \rec$ of type $k_0$ has utility such that $u^{r_v,k_0}_{\theta}=-1$ for all $\theta \in \Theta$.
		\item The sender's utility is such that, for every $\theta \in \Theta$, the function $f_\theta: 2^\rec \to [0,1]$ satisfies $f_\theta(R) =1 $ if and only if $R \subseteq \rec : |R| \geq 2$, while $f_\theta(R) =0$ otherwise.
		\item The subset $K \subseteq \K$ of type profiles that can occur with positive probability is $K \defeq \left\{  \kvec^{uv, \sigma} \mid e =(u,v) \in E , \sigma \in  \Sigma \right\}$, where, for every edge $e = (u,v) \in E$ and label $\sigma \in \Sigma$, the type profile $\kvec^{uv, \sigma} \in \K$ is such that $k^{uv, \sigma}_{r_u} = k_\sigma$, $k^{uv, \sigma}_{r_v}  = k_{\sigma'}$ with $\sigma' = \Pi_e(u)$, and $k^{uv, \sigma}_{r_{v'}}  = k_0$ for every $r_{v'} \in \rec : r_{v'} \notin \{r_u, r_v\}$.
		\item The probability distribution $\lambdavec \in \textnormal{int}(\Delta_{K})$ is such that $\lambda_{\kvec}=\frac{1}{|E||\Sigma|}$ for every $\kvec \in K$.
	\end{itemize}
	Notice that, in the \textsf{BAYESIAN-OPT-SIGNAL}  instances used for the reduction, the sender's payoff is $1$ if and only if at least two receivers play action $a_1$, while it is $0$ otherwise.
	Let us also recall that direct signals for a receiver $r_v \in \rec$ are defined by the set $\sset_{r_v} \defeq 2^{\K_{r_v}}$, with a signal being represented as the set of receiver's types that are recommended to play action $a_1$.

	\paragraph{Completeness}
	Let $\pi: U \cup V \to \Sigma$ be a labeling of the graph $G$ that satisfies all the edge constraints.
	We define a corresponding direct signaling scheme $\phi: \Theta \to \Delta_\sset$ as follows.
	For any label $\sigma \in \Sigma$, let $\svec^\sigma \in \sset$ be a signal profile such that the signal sent to receiver $r_v \in \rec$ is $s_{r_v}^\sigma = \{ k_\sigma \}$, \emph{i.e.}, only a receiver of the type $k_\sigma$ is told to play $a_1$, while all the other types are recommended to play $a_0$.
	Moreover, let $\svec^\pi \in \sset$ be a signal profile in which the signal sent to receiver $r_v \in \rec$ is $s_{r_v}^\pi = \{ k_{\sigma} \}$ with $\sigma \in \Sigma : \sigma = \pi (v)$, \emph{i.e.}, each receiver $r_v$ is told to play action $a_1$ only if her/his type is $k_{\sigma}$ for the label $\sigma$ assigned to vertex $v$ by the labeling $\pi$, otherwise she/he is recommended to play $a_0$.
	Then, we define $\phi_{\theta_\sigma} (\svec^{\sigma}) = 1$ for every state of nature $\theta_\sigma \in \Theta$, while $\phi_{\theta_0} (\svec^\pi) = 1$.
	Notice that the signaling scheme $\phi$ is deterministic, since each state of nature is mapped to only one signal profile (with probability one).
	As a first step, we prove that the signaling scheme $\phi$ is \emph{persuasive}.
	Let us fix a receiver $r_v \in \rec$.
	After receiving a signal $s = \{ k_\sigma \} \in \sset_{r_v}$ with $\sigma \in \Sigma : \sigma \neq \pi(v)$, by definition of $\phi$, the receiver's posterior belief is such that state of nature $\theta_\sigma$ is assigned probability one.
	Thus, if the receiver has type $k_\sigma$, then she/he is incentivized to play action $a_1$, since $u^{r_v,k_\sigma}_{\theta_\sigma}=\frac{1}{2} > 0$ (recall that $u^{r_v,k_\sigma}_{\theta_\sigma}$ is the utility different ``action $a_1$ minus action $a_0$'' when the state is $\theta_\sigma$).
	Instead, if the receiver has type $k \in \K_{r_v} : k \neq k_\sigma$, then she/he is incentivized to play action $a_0$, since either $k = k_0$ and $u^{r_v,k_0}_{\theta_\sigma}=-1 < 0$ or $k = k_{\sigma'}$ with $\sigma' \in \Sigma: \sigma' \neq \sigma$ and $u^{r_v,k_{\sigma'}}_{\theta_\sigma}=-1 < 0$.
	After receiving a signal $s = \{ k_{\sigma} \} \in \sset_{r_v}$ with $\sigma = \pi(v)$, the receiver's posterior belief is such that the states of nature $\theta_\sigma$ and $\theta_0$ are assigned probabilities proportional to their corresponding prior probabilities, respectively $\mu_{\theta_\sigma}$ and $\mu_{\theta_0}$ (she/he cannot tell whether $\svec^\sigma$ or $\svec^\pi$ has been selected by the sender).
	Thus, if the receiver has type $k_{\sigma}$, then she/he is incentivized to play action $a_1$, since her expected utility difference ``action $a_1$ minus action $a_0$'' is the following:
	\[
		\frac{\mu_{\theta_\sigma}}{\mu_{\theta_\sigma} + \mu_{\theta_0}} u^{r_v,k_\sigma}_{\theta_\sigma}  + \frac{\mu_{\theta_0}}{\mu_{\theta_\sigma} + \mu_{\theta_0}} u^{r_v,k_\sigma}_{ \theta_0}  = \frac{1}{\mu_{\theta_\sigma} + \mu_{\theta_0}} \left[  \frac{\epsilon}{|\Sigma|^2}\frac{1}{2}-\left( 1-\frac{\epsilon}{|\Sigma|} \right) \frac{\epsilon}{2|\Sigma|^2} \right] >\frac{1}{\mu_{\theta_\sigma} + \mu_{\theta_0}} \left[ \frac{\epsilon}{2|\Sigma|^2} -\frac{\epsilon}{2|\Sigma|^2} \right] =0. 
	\]
	If the receiver has a type different from $k_\sigma$, simple arguments show that the expected utility difference is negative, incentivizing action $a_0$.
	This proves that the signaling scheme $\phi$ is persuasive.
	Next, we bound the sender's expected utility in $\phi$.
	Notice that, when the state of nature is $\theta_0$, if the receivers' type profile is $\kvec^{uv, \sigma} \in K$ with $\sigma = \pi(u)$ for some edge $e = (u,v) \in E$, then both receivers $r_u$ and $r_v$ play action $a_1$.
	This is readily proved since $k^{uv, \sigma}_{r_u} = k_\sigma$ and $k^{uv, \sigma}_{r_v}  = k_{\sigma'}$ with $\sigma = \pi(u)$ and $\sigma' = \pi(v)$ (recall that $\pi(v) = \Pi_e(u)$ as $\phi$ satisfies all the edge constraints), and, thus, both $r_u$ and $r_v$ are recommended to play $a_1$ when the state is $\theta_0$.
	As a result, under signaling scheme $\phi$, when the receivers' type profile is $\kvec^{uv, \sigma} \in K$, then the sender's resulting payoff is one (recall the definition of functions $f_\theta$).
	By recalling that each type profile $\kvec^{uv, \sigma} \in K$ with $\sigma = \pi(u)$ (for each edge $e = (u,v) \in E$) occurs with probability $\lambda_{\kvec^{uv, \sigma}} = \frac{1}{|E||\Sigma|}$, we can lower bound the sender's expected utility (see the objective of Problem~\eqref{LP1}) as follows:
	\[
		\sum_{\kvec \in K} \lambda_{\kvec} \sum_{\theta\in\Theta} \mu_\theta \sum_{\svec \in \sset}\phi_\theta(\svec) f_\theta \left( R_{\svec}^{\kvec} \right) \geq \mu_{\theta_0} \sum_{\kvec^{uv, \sigma} \in K : \sigma = \pi(u)} \lambda_{\kvec^{uv, \sigma}} = \mu_{\theta_0} \frac{1}{|\Sigma|} = \left( 1-\frac{\epsilon}{|\Sigma|} \right) \frac{1}{|\Sigma|}.
	\]

	\paragraph{Soundness}
	By contradiction, suppose that there exists a direct and persuasive signaling scheme $\phi: \Theta \to \Delta_\sset$ that provides the sender with an expected utility greater than $\frac{2 \epsilon}{|\Sigma|}$.
	Since the sender can extract an expected utility at most of $\frac{\epsilon}{|\Sigma|}$ from states of nature $\theta \in \Theta$ with $\theta \neq \theta_0$ (as $\sum_{\theta \in \Theta : \theta \neq \theta_0} \mu_\theta = \frac{\epsilon}{|\Sigma|}$ and the maximum value of functions $f_\theta$ is one), then it must be the case that the expected utility contribution due to state $\theta_0$ is greater than $\frac{\epsilon}{|\Sigma|}$.
	Let us consider the distribution over signal profiles $\phi_{\theta_0} \in \Delta_\sset$ induced by state of nature $\theta_0$.
	We prove that, for each signal profile $\svec \in \sset$ such that $\phi_{\theta_0}(\svec) > 0$ and each receiver $r_v \in \rec$, it must hold that $|s_r| \leq 1$, \emph{i.e.}, at most one type of receiver $r_v$ is recommended to play $a_1$.
	First, notice that a receiver of type $k_0$ cannot be incentivized to play $a_1$, since $u^{r_v,k_0}_{\theta}=-1$ for all $\theta \in \Theta$.
	By contradiction, suppose that there are two receiver's types $k_{\sigma}, k_{\sigma'} \in \K_{r_v}$ with $k_{\sigma} \neq k_{\sigma'} $ such that $k_{\sigma}, k_{\sigma'}  \in s_r$ (\emph{i.e.}, they are both recommended to play $a_1$).
	By letting $\boldsymbol{\xi} \in \Delta_\Theta$ be the posterior belief of receiver $r_v$ induced by $s_r$, for type $k_\sigma$ it must be the case that:
	\[
		\xi_{\theta_\sigma} u^{r_v,k_\sigma}_{\theta_\sigma}  + \sum_ {\theta_{\sigma''} \in \Theta : \theta_{\sigma''} \neq \theta_{\sigma}  }  \xi_{\theta_{\sigma''}} u^{r_v,k_\sigma}_{\theta_{\sigma''}} + \xi_{\theta_0} u^{r_v,k_\sigma}_{\theta_0} = \frac{1}{2} \xi_{\theta_\sigma} - \sum_ {\theta_{\sigma''} \in \Theta : \theta_{\sigma''} \neq \theta_{\sigma}  }  \xi_{\theta_{\sigma''}} - \frac{\epsilon}{2 |\Sigma|^2} \xi_{\theta_0} > 0,
	\]
	since the signaling scheme is persuasive, and, thus, a receiver of type $k_\sigma$ must be incentivized to play action $a_1$.
	This implies that $\xi_{\theta_\sigma} > 2 \sum_ {\theta_{\sigma''} \in \Theta : \theta_{\sigma''} \neq \theta_{\sigma}  }   \xi_{\theta_{\sigma''}} \ge 2 \xi_{\theta_{\sigma'}}$.
	Analogous arguments for type $k_{\sigma'}$ imply that $\xi_{\theta_{\sigma'}} > 2 \xi_{\theta_{\sigma}}$, reaching a contradiction.
	This shows that, for each $\svec \in \sset$ such that $\phi_{\theta_0}(\svec) > 0$ and each $r_v \in \rec$, it must be the case that $|s_r| \leq 1$.
	Next, we provide the last contradiction proving the result.
	Let us recall that, by assumption, the sender's expected utility contribution due to $\theta_0$ is $\sum_{\kvec \in K} \lambda_{\kvec} \sum_{\svec \in \sset}\phi_{\theta_0}(\svec) f_{\theta_0} \left( R_{\svec}^{\kvec} \right) \geq \frac{\epsilon}{|\Sigma|}$.
	By an averaging argument, this implies that there must exist a signal profile $\svec \in \sset$ such that $\phi_{\theta_0}(\svec) > 0$ and $\sum_{\kvec \in K} \lambda_{\kvec}  f_{\theta_0} \left( R_{\svec}^{\kvec} \right) \geq \frac{\epsilon}{|\Sigma|}$.
	Let $\svec \in \sset$ be such signal profile.
	Let us define a corresponding labeling $\pi: U \cup V \to \Sigma$ of the graph $G$ such that, for every vertex $v \in U \cup V$, it holds $\pi(v) = \sigma$, where $\sigma \in \Sigma$ is the label corresponding to the unique type $k_\sigma$ of receiver $r_v$ that is recommended to play action $a_1$ under $\svec$ (if any, otherwise any label is fine). 
	Since $\sum_{\kvec \in K} \lambda_{\kvec}  f_{\theta_0} \left( R_{\svec}^{\kvec} \right) \geq \frac{\epsilon}{|\Sigma|}$ and it holds $\lambda_{\kvec} = \frac{1}{|E||\Sigma|}$ and $f_{\theta_0} \left( R_{\svec}^{\kvec} \right) \in \{0,1\}$ for every $\kvec \in K$, it must be the case that there are at least $\epsilon |E|$ type profiles $\kvec \in K$ such that $f_{\theta_0} \left( R_{\svec}^{\kvec} \right) = 1$.
	Since a receiver of type $k_0$ cannot be incentivized to play action $a_1$, the value of $f_{\theta_0} \left( R_{\svec}^{\kvec} \right) $ can be one only if there are at least two receivers with types different from $k_0$ that play action $a_1$.
	Thus, it must hold that $f_{\theta_0} \left( R_{\svec}^{\kvec} \right) = 0$ for all the type profiles $\kvec^{uv, \sigma} \in K$ such that $ \sigma \neq \pi(u)$ (as $k^{uv, \sigma}_{r_u}$ would be equal to $k_\sigma$ with $\sigma \neq \pi(u)$ and $k_\sigma \notin s_{r_u}$).
	For the type profiles $\kvec^{uv, \sigma} \in K$ such that $ \sigma = \pi(u)$ (one per edge $e = (u,v) \in E$ of the graph $G$), the value of $f_{\theta_0} \left( R_{\svec}^{\kvec} \right) $ is one if and only if $\pi(v) = \Pi_e(u)$, so that both receivers $r_u$ and $r_v$ are told to play action $a_1$.
	As a result, this implies that there must be at least $\epsilon |E|$ edges $e \in E$ for which the labeling $\pi$ satisfies the corresponding edge constraint $\Pi_e$, which is a contradiction.
\end{proof}

\section{Proofs Omitted from Section~\ref{sec:learning_oracle}}

\thmregretogd*

\begin{proof}
	First, we bound the per-iteration running time of Algorithm~\ref{alg:OGD}.
	For any $t \in [T]$, we have $E^{t} = \bigcup_{t' \in [t]} e^{t'}$, which represents the set of feedbacks observed up to iteration $t$.
	Thus, it holds $|E^{t}| \leq t$.
	At iteration $t \in [T]$, the algorithm works with vectors $\xvec^t$ and $\yvec^{t+1}$.
	The first one belongs to $\X_{E^{t-1}}$ (as it is returned by $\varphi_\alpha$ at iteration $t-1$), and, thus, it has at most $t-1$ non-zero components.
	Similarly, since $\yvec^{t+1} = \xvec^t+\eta \ind_{e^t} $, it holds that $\yvec^{t+1} \in [0,2]^p$ and $y^{t+1}_e = 0$ for all $e \notin E^{t}$, which implies that $\yvec^{t+1}$ has at most $t$ non-zero components.
	As a result, we can sparsely represent vectors $\xvec^t$ and $\yvec^{t+1}$ so that Algorithm~\ref{alg:OGD} has a per-iteration running time bounded by $t$ for any iteration $t \in [T]$, independently of the actual size $p$ of the vectors.
	Moreover, notice that $\yvec^{t+1}$ satisfies the conditions required by the inputs of the oracle $\varphi_\alpha$.
	%
	%

	Next, we bound the $\alpha$-regret of Algorithm~\ref{alg:OGD}.
	For the ease of notation, in the following, for any vector $\xvec \in \X$ and subset $E \subseteq \E$, we let $\xvec_{E} \defeq \tau_{E}(\xvec)$.
	Moreover, for any $t \in [T]$, we let $\mathbb{I}_t \defeq \mathbb{I} \left\{ e^{t}\notin E^{t-1} \right\}$, which is the indicator function that is equal to $1$ if and only if $e^{t}\notin E^{t-1}$, \emph{i.e.}, when the feedback $e^t$ at iteration $t$ has never been observed before.
	Fix $\xvec \in \alpha \X$.
	Then, the following relations hold:
	\begin{subequations}
		\begin{align}
		\big\vert\big\vert \xvec_{E^{t}}-\xvec^{t+1} \big\vert\big\vert^2 & \le \big\vert\big\vert \xvec_{E^t}-\yvec^{t+1}\vert\vert^2+ \epsilon \\
		& = \big\vert\big\vert \xvec_{E^t}-\xvec^{t} -\eta \ind_{e^t} \big\vert\big\vert^2+\epsilon \label{eq:x_1}\\
		& = \big\vert\big\vert \xvec_{E^{t-1}} +\mathbb{I}_t \, x_{e^t} \ind_{e^t} - \xvec^{t} -\eta \ind_{e_t} \big\vert\big\vert^2+\epsilon \\
		& = \big\vert\big\vert \xvec_{E^{t-1}} +\mathbb{I}_t \, x_{e^t} \ind_{e^t}  -\xvec^{t} \big\vert\big\vert^2 + \eta^2 - 2 \eta \ind_{e^t}^\top  \Big( \xvec_{E^{t-1}} + \mathbb{I}_t \, x_{e^t} \ind_{e^t} -\xvec^{t} \Big) +\epsilon \label{eq:x_2} \\
		& = \big\vert\big\vert \xvec_{E^{t-1}} -\xvec^t_{E^{t-1}} \big\vert\big\vert^2 + \mathbb{I}_t \, \big| x_{e^{t}}-x^t_{e^t} \big|^2+\eta^2 - 2 \eta \ind_{e^t}^\top  \Big( \xvec_{E^{t-1}} + \mathbb{I}_t \, x_{e^t} \ind_{e^t} -\xvec^{t} \Big) +\epsilon  \label{eq:x_3}\\
		& \leq \big\vert\big\vert \xvec_{E^{t-1}}-\xvec^t_{E^{t-1}} \big\vert\big\vert^2 + \mathbb{I}_t+\eta^2 - 2 \eta \ind_{e^t}^\top  \Big( \xvec_{E^{t-1}} + \mathbb{I}_t \, x_{e^t} \ind_{e^t} -\xvec^{t} \Big) +\epsilon.
		\end{align}
	\end{subequations}
	Notice that Equation~\eqref{eq:x_1} holds by definition of $\varphi_\alpha$ since $\xvec_{E^{t}} \in \alpha \X_{E^t}$, Equation~\eqref{eq:x_2} follows from $\xvec_{E^{t}}  = \xvec_{E^{t-1}} +\mathbb{I}_t \, x_{e^t} \ind_{e^t} $, while Equation~\eqref{eq:x_3} can be derived by decomposing the first squared norm in the preceding expression.
	By using the last relation above, we can write the following:
	\begin{subequations}
		\begin{align}
		\sum_{t \in [T]} \ind_{e^t} ^\top \Big( \xvec - \xvec^{t} \Big) & = \sum_{t \in [T]} \ind_{e^t}^\top \Big( \xvec_{E^{t-1}} + \mathbb{I}_t \, x_{e^t} \ind_{e^t} -\xvec^{t} \Big) \\
		& \le \frac{1}{2\eta}\sum_{t \in [T]}\Bigg(  \big\vert\big\vert \xvec_{E^{t-1}}-\xvec^t_{E^{t-1}} \big\vert\big\vert^2 - \big\vert\big\vert \xvec_{E^{t}}-\xvec^{t+1} \big\vert\big\vert^2 + \mathbb{I}_t +\eta^2 + \epsilon \Bigg) \\
		&= \frac{1}{2\eta}\sum_{t \in [T]} \Big( \mathbb{I}_t +\eta^2+\epsilon \Big)  \label{eq:x_4}\\
		&= \frac{1}{2\eta} \Big( | E^T |+T\eta^2+T\epsilon \Big), 
		\end{align}
	\end{subequations}
	where Equation~\eqref{eq:x_4} is obtained by telescoping the sum.
	Then, the following concludes the proof:
	\begin{align*}
		R^T_\alpha & \defeq \alpha \max_{y \in \Y} \sum_{t \in [T]} u(y,e^t) - \sum_{t \in [T]} u(y^t,e^t) \leq  \alpha \max_{\xvec \in \X}  \sum_{t \in [T]}  x_{e^t} - \sum_{t \in [T]} x^t_{e^t} = \alpha \max_{\xvec \in \X}  \sum_{t \in [T]}  \ind_{e^t}^\top \left( \xvec - \xvec^{t} \right)  \\
		& =  \max_{\xvec \in \alpha\X}  \sum_{t \in [T]} \ind_{e^t}^\top \left( \xvec - \xvec^{t} \right)  \leq \frac{1}{2\eta} \Big( | E^T |+T\eta^2+T\epsilon \Big).
	\end{align*}
\end{proof}

\section{Proofs Omitted from Section~\ref{sec:offline}}
\thmOffline*
\begin{proof}[Proof of~\cref{th:bayesian}]
The dual problem of LP~\eqref{LP1} reads as follows:
\begin{subequations} \label{LP:dual}
	\begin{align}\min_{\zvec,\dvec}  & \quad \sum_{\theta\in\Theta}d_\theta\\
	\textnormal{s.t. } & \mu_\theta \sum_{r \in \rec} \sum_{k\in s_r}u^{r,k}_\theta z_{r,s_r,k}+d_\theta \ge  \mu_\theta \sum_{\kvec \in  \K} \lambda_{\kvec} f_\theta(R_{\svec}^{\kvec}) & \forall \theta \in \Theta, \forall \svec \in \Scal \label{DUAL:incentive} \\
	&z_{r,s,k}\le 0 &\forall r \in \rec, \forall s \in \Scal_r, \forall k \in \K_r:k\in s,
	\end{align} 
\end{subequations}
where $\dvec\in\mathbb{R}^{|\Theta|}$ is the vector of dual variable corresponding to the primal Constraints~\eqref{LP1:constr_simplex}, and $\zvec\in\mathbb{R}_{-}^{|\rec\times\Scal_r\times \K_r|}$ is the vector of dual variable corresponding to Constraints~\eqref{LP1:cons_pers} in the primal.
We rewrite the dual LP~\eqref{LP:dual} so as to highlight the relation between an approximate separation oracle for Constraints~\eqref{DUAL:incentive} and the oracle $\Ocal_{\alpha}$. Specifically, we have
\begin{subequations} \label{LP:dual2}
	\begin{align}
	\min_{\zvec\ge 0,\dvec} & \quad \sum_{\theta\in\Theta}d_\theta\\
	\textnormal{s.t. } & d_\theta \ge \mu_\theta \mleft(\sum_{\kvec \in \K} \lambda_{\kvec} f_\theta(R_{\svec}^{\kvec}) +  \sum_{r \in \rec} \sum_{k \in s_r} u^{r,k}_\theta z_{r,s_r,k}\mright) & \forall \theta \in \Theta,\forall \svec \in \Scal. \label{eq:dual2const}
	\end{align} 
\end{subequations}

Now, we show that it is possible to build a binary search scheme to find a value $\gamma^\star\in [0,1]$ such that the dual problem with objective $\gamma^\star$ is feasible, while the dual with objective $\gamma^\star-\beta$ is infeasible. The constant $\beta \geq 0$ will be specified later in the proof. The algorithm requires $\log(\beta)$ steps and works by determining, for a given value $\bar\gamma \in [0,1]$, whether there exists a feasible pair $(\dvec,\zvec)$ for the following feasibility problem \circled{F}:
\begin{equation*}
\circled{F} \, \, \mleft\{\hspace{-1.25mm}\begin{array}{l}
\displaystyle
\sum_{\theta\in\Theta}d_\theta\le\bar\gamma\\
\displaystyle d_\theta \ge \mu_\theta \mleft(\sum_{\kvec \in \K} \lambda_{\kvec} f_\theta(R_{\svec}^{\kvec}) +  \sum_{r \in \rec} \sum_{k \in s_r} u^{r,k}_\theta z_{r,s_r,k}\mright)\hspace{1cm}\forall \theta \in \Theta,\forall \svec \in \Scal\\
\zvec\ge 0.
\end{array}\mright.
\end{equation*}
At each iteration of the bisection algorithm, the feasibility problem \circled{F} is solved via the ellipsoid method.
The algorithm is inizialized with $l=0$, $h=1$, and $\bar\gamma=\frac{1}{2}$.
If \circled{F} is infeasible for $\bar\gamma$, the algorithm sets $l \gets (l+h)/2$ and $\bar\gamma\gets (h+\bar\gamma)/2$. Otherwise, if \circled{F} is (approximately) feasible, it sets $h\gets (l+h)/2$ and $\bar \gamma\gets (l+\bar\gamma)/2$. Then, the procedure is repeated with the updated value of $\bar\gamma$. The bisection procedure terminates when it determines a value $\gamma^\star$ such that \circled{F} is feasible for $\bar\gamma=\gamma^\star$, while it is infeasible for $\bar\gamma=\gamma^\star-\beta$. 
In the following, we present the approximate separation oracle which is employed at each iteration of the ellipsoid method.

\paragraph{Separation Oracle}
Given a point $(\bar\dvec,\bar\zvec)$ in the dual space, and $\bar\gamma \in[0,1]$, we design an approximate separation oracle to determine if the point $(\bar\dvec,\bar\zvec)$ is approximately feasible, or to determine a constraint of \circled{F} that is violated by such point.
For each $\theta\in\Theta$, $r\in\rec$, and $s\in\Scal_r$, let
\[
w^\theta_{r,s}\defeq\mu_\theta \sum_{k \in s} u^{r,k}_\theta \bar z_{r,s,k}.
\]
When the magnitude of the weights $|w^\theta_{r,s}|$ is small, we show that it is enough to employ the optimization oracle $\Ocal_\alpha$ in order to find a violated constraint, or to certify that all the constraints are approximately satisfied. On the other hand, when the weights $|w^\theta_{r,s}|$ are large (in particular, when the largest weight has exponential size in the size of the problem instance), the optimization oracle $\Ocal_\alpha$ loses its polynomial time guarantees (see Definition~\ref{def:oracle}). We show how to handle those specific settings in the following case analysis:

\begin{itemize}

\item Equation~\eqref{eq:dual2const} implies that $d_\theta\ge 0$ for each $\theta \in \Theta$. Then, if there exists a $\theta \in \Theta$ such that $\bar d_{\theta}<0$, we return the violated constraint $(\theta, \varnothing)$ (that is, $d_\theta \ge 0$).

\item If there exists $\theta \in \Theta$ such that $\bar d_\theta > 1$, then the first constraint of \circled{F} must be violated as $\bar\gamma\in[0,1]$.

\item If there exists a receiver $r \in \rec$ and a signal $s\in\Scal_r$ such that $w^\theta_{r,s}>1$, then the constraint of \circled{F} corresponding to the pair $(\theta,s)$ is violated, because $d_\theta\le 1$.

\item  If no violated constraint was found in the previous steps, we proceed by checking if there exists a state $\theta' \in \Theta$, a receiver $r' \in \rec$, and a signal $s'\in \Scal_r$, such that $w^{\theta'}_{r',s'}\le - |\rec|$. If this is the case, we observe that for any pair $(\theta',\svec)$, with $\svec\in\Scal : s_r=s'$, the corresponding constraint in \circled{F} reads
\[
\mu_\theta \sum_{\kvec \in \K} \lambda_{\kvec} f_\theta(R^{\kvec}_{\svec}) +  \sum_{r \in \rec\setminus \{r'\}} w^{\theta'}_{r,s_r} +w^{\theta'}_{r',s'} \le 0,
\] 
since $\bar \dvec\ge 0$ if the current step is reached. For $w^{\theta'}_{r',s'}\le - |\rec|$ the above constraints are trivially satisfied, and therefore we can safely manage (for the current iteration of the ellipsoid method) any such constraint by setting $w^{\theta'}_{r',s'}= - |\rec|$.

\end{itemize}

If none of the previous steps returned a violated constraint, we can safely assume that $0\le d_\theta \le 1$ and $-|\rec|\le w^\theta_{r,s} \le 1$, for each $\theta\in\Theta$, $r\in\rec$, and $s\in\Scal_r$. Moreover, we observe that, by definition, for each $r \in \rec$ and $\theta \in \Theta$, it holds $w^\theta_{r,\varnothing}=0$.
Since the magnitude of the weights is guaranteed to be small (that is, weights are guaranteed to be in the range $[-|\rec|,1]$), for each $\theta \in \Theta$ we can invoke $\Ocal_{\alpha}(\theta,\K,\lambda,\wvec^\theta,\delta)$ to determine an $\svec^\theta \in \sset$ such that 
\[
\mu_\theta \sum_{\kvec \in \K} \lambda_{\kvec} f_\theta(R^{\kvec}_{\svec^\theta}) + \sum_{r \in \rec} w^\theta_{r,s^\theta_r} \ge  \max_{\svec\in \Scal} \mleft\{ \alpha \mu_\theta \sum_{\kvec\in\K} \lambda_{\kvec} f_\theta(R_{\svec}^{\kvec}) + \sum_{r \in \rec}  w^\theta_{r,s_r}\mright\} -\delta,
\]
where $\delta$ is an approximation error that will be defined in the following.
If at least one $\svec^\theta$ is such that $(\theta,\svec^\theta)$ is violated, we output that constraint, otherwise the algorithm returns that the LP is feasible.

\paragraph{Putting It All Together}
The bisection algorithm computes a $\gamma^\star\in [0,1]$ and a pair $(\bar \dvec, \bar\zvec)$ such that the approximate separation oracle does not find a violated constraint.
The following lemma defines a modified LP and shows that $(\bar \dvec, \bar\zvec)$ is a feasible solution for this problem and has value at most $\gamma^\star$.

\begin{lemma} \label{lm:violated}
	The pair $(\bar \dvec, \bar\zvec)$ is a feasible solution to the following LP and has value at most $\gamma^\star$:
	\begin{subequations} \label{LP:dual3}
		\begin{align*}
			\min_{\zvec\ge 0,\dvec} & \quad  \sum_{\theta\in\Theta}d_\theta \\
			\textnormal{s.t. } & d_\theta \ge \, \alpha \mu_\theta \sum_{\kvec \in \K} \lambda_{\kvec} f_\theta(R_{\svec}^{\kvec}) + \mu_\theta \sum_{r \in \rec} \sum_{k \in s_r} u^{r,k}_\theta z_{r,s_r,k}\ -\delta & \forall \theta \in \Theta,\forall \svec \in \Scal.
		\end{align*} 
	\end{subequations}
\end{lemma}

\begin{proof}
The value is at most $\gamma^\star$ by assumption (that is, the separation oracle does not find a violated constraint for $(\bar \dvec, \bar\zvec)$ in \circled{F} with objective $\gamma^\star$).
Analogously, it holds that $\bar d_\theta \in [0,1]$ for each $\theta \in \Theta$, and $w^\theta_{r,s}\le1$ for each $r \in \rec$, $s \in \sset_r$, and $\theta \in \Theta$.
Suppose, by contradiction, that $(\theta,\svec')$ is a violated constraint of the modified LP above. Then, given $\bar\dvec$, oracle $\Ocal_\alpha$ would have found an $\svec \in \sset$ such that 
\[ 
\mu_\theta \sum_{\kvec \in \K} \lambda_{\kvec} f(R_{\svec}^{\kvec}) + \mu_\theta \sum_{r \in \rec} \sum_{k \in s_r} u^{r,k}_\theta \bar z_{r,s_r,k} \ge \alpha \sum_{\theta\in\Theta} \,\, \mu_\theta \sum_{\kvec \in \K} \lambda_{\kvec} f_\theta(R_{\svec'}^{\kvec}) + \mu_\theta \sum_{r \in \rec} \sum_{k \in s'_r} u^{r,k}_\theta \bar z_{r,s'_r,k} -\delta > \bar d_\theta,
\] 
where the first inequality follows by Definition~\ref{def:oracle}, and the second from the assumption that the modified dual is infeasible. Hence, $\mathcal{O_{\alpha}}$ would return a violated constraint, reaching a contradiction.
\end{proof}

The dual problem of the LP of Lemma~\ref{lm:violated} reads as follows:
\begin{subequations}\label{LP2}
	\begin{align*}  
	\max_\phi & \sum_{\svec \in \Scal}\ \sum_{\theta\in\Theta} \phi_\theta(\svec) \mleft(\alpha \mu_\theta \sum_{\kvec \in \K} \lambda_{\kvec}  \, f_\theta(R_{\svec}^{\kvec})-\delta\mright) \\
	\textnormal{s.t. } & \sum_{\theta\in\Theta}\mu_\theta \sum_{\svec: s_r=s'} \,\phi_\theta(\svec) \, u^{r,k}_\theta \geq 0 & \forall r \in R, \forall s' \in \Scal_r, \forall k \in \K_r: k\in s' \\
	&\sum_{\svec \in\Scal}\phi_{\theta}(\svec)= 1 & \forall \theta\in\Theta\\
	&\phi_\theta(\svec)\ge 0 & \forall \theta \in \Theta, \svec \in \Scal.
	\end{align*} 
\end{subequations}
By strong duality, Lemma \ref{lm:violated} implies that the value of the above problem is at most $\gamma^\star$.
Then, let $\opt$ be value of the optimal solution to LP~\eqref{LP1}. The same solution is feasible for the LP we just described, where it has value 
\begin{equation}\label{alfaopt}
\alpha \opt-|\Theta| \delta \le \gamma^\star.
\end{equation}

Now, we show how to find a solution to the original problem (LP~\eqref{LP1}) with value at least $\gamma^\star-\beta$.
Let $\Hcal$ be the set of constraints returned by the ellipsoid method run on the feasibility problem \circled{F} with objective $\gamma^\star-\beta$. 

\begin{lemma}\label{lp1restricted}
	 LP~\eqref{LP1} with variables restricted to those corresponding to dual constraints $\Hcal$ returns a signaling scheme with value at least $\gamma^\star-\beta$. Moreover, the solution can be determined in polynomial time.
\end{lemma}
\begin{proof}
By construction of the bisection algorithm, \circled{F} is infeasible for value $\gamma^\star-\beta$.
Hence, the following LP has value at least $\gamma^\star-\beta$:
\begin{subequations}
	\begin{align*}
	\min_{\zvec\ge 0,\dvec} & \quad \sum_{\theta\in\Theta}d_\theta\\
	\textnormal{s.t. } & \,\,d_\theta \ge \mu_\theta \mleft(\sum_{\kvec \in \K} \lambda_{\kvec} f_\theta(R_{\svec}^{\kvec}) +  \sum_{r \in \rec} \sum_{k \in s_r} u^{r,k}_\theta z_{r,s_r,k}\mright)& \forall(\theta,\svec) \in \Hcal.
	\end{align*} 
\end{subequations}
%
Notice that the primal of the above LP is exactly LP \eqref{LP1} with variables restricted to those corresponding to dual constraints in $\Hcal$, and that the former (restricted) LP has value at least $\gamma^\star- \beta$ by strong duality. To conclude the proof, the ellipsoid method guarantees that $\Hcal$ is of polynomial size. Hence, the LP can be solved in polynomial time.
\end{proof}

Let $\apx$ be the value of an optimal solution to LP~\eqref{LP1} restricted to variables corresponding to dual constraints in $\Hcal$. Then, 
\begin{align*}
\apx & \ge \gamma^\star - \beta \\
& \ge \alpha \opt - |\Theta|\delta -\beta\\
& \ge \alpha \opt - \epsilon,
\end{align*}
where the first inequality holds by Lemma~\ref{lp1restricted}, the second inequality follows from Equation~\eqref{alfaopt}, and the last inequality is obtained by setting $\delta=\frac{\epsilon}{2|\Theta|}$ and $\beta=\frac{\epsilon}{2}$.
%
\end{proof}

\section{Proofs Omitted from Section~\ref{sec:sep_to_proj}}

\thmProjConstruction*

\begin{proof}

The problem of computing the projection of point $\yvec$ on $\X_{K}$ (see Equation~\eqref{eq:def_x}) can be formulated via the following convex programming problem, which we denote by \circled{P}:
\begin{equation*}
\circled{P}\mleft\{ \,\, \begin{array}{l}
\displaystyle \min_{\phi, \xvec} \hspace{.1cm} \sum_{\kvec \in K} (x_{\kvec}-y_{\kvec})^2\\[3mm]
\displaystyle \textnormal{s.t. }\hspace{.2cm} \sum_{\theta\in\Theta}\mu_\theta \mleft(\sum_{\substack{\svec\in\Scal:\\ s_r=s'}} \,\phi_\theta(\svec) u^{r,k}_\theta\mright) \geq 0  \hspace{1cm} \forall r \in \rec,\forall s'\in \sset_r ,\forall k \in \K_r:k\in s'\\[7mm]
\displaystyle \hspace{.8cm} \sum_{\svec\in  \Scal}\phi_\theta(\svec)= 1 \hspace{3.4cm} \forall\theta\in\Theta \\[5mm]
\displaystyle \hspace{.8cm} \phi_{\theta}(\svec)\ge 0 \hspace{4cm} \forall \theta \in \Theta, \forall \svec \in \Scal\\[2mm]
\displaystyle \hspace{.8cm} x_{\kvec} \le  \sum_{\theta \in \Theta} \sum_{\svec \in \Scal}\,\mu_\theta\,\phi_\theta(\svec) f_\theta(R_{\svec}^{\kvec}) \hspace{1.1cm} \forall \kvec \in K.
\end{array}\mright. .
\end{equation*}

Then, we compute the Lagrangian of~\circled{P} by introducing dual variables $z_{r,s,k} \le 0$ for each $r\in\rec,s\in\Scal_r,$ and $k\in s$, $d_\theta \in \mathbb{R}$ for each $\theta\in\Theta$,  $v_{\theta,\svec}\le 0$  for each $\theta\in\Theta$, $\svec\in\Scal$, and $\nu_{\kvec} \ge 0$ for each $\kvec \in K$. Specifically, the Lagrangian of \circled{P} reads as follows
\begin{equation*}
\begin{array}{l}
\displaystyle 	L(\phi,\xvec,\zvec, \vvec, \bb{\nu}, \dvec)\defeq \sum_{\kvec \in K} (x_{\kvec} -y_{\kvec})^2 +
\sum_{ r \in \rec}\sum_{s' \in \sset_r}\sum_{k\in s'} z_{r,s,k}\mleft( \sum_{\theta\in\Theta}\mu_\theta \sum_{\svec: s_r=s'} \,\phi_\theta(\svec) \,u^{r,k}_\theta \mright)\\[6mm]
\displaystyle\hspace{4.3cm}+\sum_{\theta\in\Theta, \svec \in \Scal} v_{\theta,\svec} \phi_\theta(\svec)+ \sum_{\theta \in \Theta} d_\theta \mleft(\sum_{\svec\in  \Scal}\phi_\theta(\svec)-1\mright)\\[6mm]
\displaystyle\hspace{4.3cm}+ \sum_{\kvec \in K} \nu_{\kvec} \mleft(x_{\kvec}- \sum_{\theta\in\Theta,\svec \in \Scal}   \mu_\theta \phi_\theta(\svec) f_\theta(R_{\svec}^{\kvec})\mright).
\end{array}
\end{equation*}
We observe that Slater's condition holds for \circled{P} (all constraints are linear, and by setting $\xvec=0$ any signaling scheme $\phi$ constitutes a feasible solution). Therefore, by strong duality, an optimal dual solution must satisfy the KKT conditions. In particular, in order for stationarity to hold, it must be $\boldsymbol{0}\in\partial_{\phi_\theta(\svec)}(L)$ for each $\svec$ and $\theta$. Then, for each $\theta\in\Theta$ and $\svec\in\Scal$, we have
\begin{equation*}
	\partial_{\phi_\theta(\svec)}(L)=\sum_{ r \in \rec }\sum_{k\in s_r} \mu_\theta z_{r,s_r,k}  u^{r,k}_\theta+  v_{\theta,\svec}+
	d_\theta -\sum_{\kvec \in K}  \nu_{\kvec} \mu_\theta  f_\theta(R_{\svec}^{\kvec})=0.
\end{equation*}
Then, for each $\theta \in \Theta$ and $\svec \in \sset$, we obtain
\begin{equation}\label{constr1}
	\sum_{ r \in \rec}\sum_{k\in s_r } \mu_\theta  z_{r,s_r,k} u^{r,k}_\theta+	d_\theta -\sum_{\kvec \in K}  \nu_{\kvec} \mu_\theta  f_\theta(R_{\svec}^{\kvec}))\ge0.
\end{equation}
Moreover, stationarity has to hold with respect to variables $\xvec$. Formally, for each $\kvec\in K$, 
\[
	\partial_{x_{\kvec}}(L) = 2(x_{\kvec} -y_{\kvec}) \nu_{\kvec} =0.
\]
Therefore, for each $\kvec\in K$,
\begin{equation}\label{constr2}
	 x_{\kvec} = y_{\kvec} -\frac{\nu_{\kvec}}{2}.
\end{equation}

By Equations~\eqref{constr1} and~\eqref{constr2}, we obtain the following dual quadratic program
\begin{equation*}
\circled{D} \,\, \mleft\{\hspace{-1.25mm}\begin{array}{l}
\displaystyle \max_{\zvec,\vvec,\bb{\nu},\dvec}  \,\,\sum_{\kvec \in K} \mleft(\nu_{\kvec} y_{\kvec}-\frac{\nu_{\kvec}^2}{4}\mright)	-\sum_{\theta\in\Theta} d_\theta \\[5mm]
\textnormal{ s.t. }  \displaystyle \hspace{.5cm} d_\theta \ge \sum_{\kvec \in K}  \nu_{\kvec} \mu_\theta  f_\theta(R_{\svec}^{\kvec})+ \sum_{ r \in \rec}\sum_{k\in s_r} \mu_\theta z_{r,s_r,k} u^{r,k}_\theta \hspace{1cm} \forall \theta \in \Theta, \forall \svec \in \Scal\\[5mm]
\displaystyle \hspace{1.1cm} z_{r,s,k} \ge 0 \hspace{6.5cm} \forall r \in \rec, \forall s \in \sset_r, \forall k \in \mathcal{K}_r: k \in s \\[3mm]
\displaystyle \hspace{1.15cm} \nu_{\kvec} \ge 0 \hspace{6.85cm} \forall \kvec \in  K,
\end{array}\mright.
\end{equation*} 
in which the objective function is obtained by observing that each term $\phi_\theta(\svec)$ in the definition of $L$ is multiplied by $\partial_{\phi_\theta(\svec)}(L)$, which has to be equal to zero by stationarity.
Similarly to what we did in the proof of Theorem~\ref{th:bayesian}, we repeatedly apply the ellipsoid method equipped with an approximate separation oracle to problem \circled{D}.
In this case, the analysis is more involved than what happens in Theorem~\ref{th:bayesian}, because we are interested in computing an approximate projection on $\alpha \X_{K}$ rather than an approximate solution of \circled{P}.
We proceed by casting \circled{D} as a feasibility problem with a certain objective (analogously to \circled{F} in Theorem~\ref{th:bayesian}). 
In particular, given objective $\gamma\in[0,1]$, the objective function of \circled{D} becomes the following constraint in the feasibility problem
\begin{equation}\label{eq:feasibility_proj}
\sum_{\kvec \in K} \mleft(\nu_{\kvec} y_{\kvec}-\frac{\nu_{\kvec}^2}{4}\mright)	-\sum_{\theta\in\Theta} d_\theta\ge \gamma.
\end{equation}
Then, given an approximation oracle $\Ocal_\alpha$ which will be specified later, we apply to the feasibility problem the search algorithm described in Algorithm~\ref{alg:projection}.
\begin{algorithm}[H]\caption{\textsc{Search Algorithm}}
	\textbf{Input:} Error $\epsilon$, $\yvec \in \mathbb{R}_+^{|\K|}$, subspace $K \subseteq \K$.
	\begin{algorithmic}[1]
		\STATE \textbf{Initialization}: $\beta \gets \frac{\epsilon}{2}$, $\delta \gets \frac{\epsilon}{2|\Theta|}$, $\gamma \gets |K|+\beta$, and $\Hcal\gets\varnothing$.
		\REPEAT
		\STATE $\gamma \gets \gamma-\beta$ 
		\STATE $\Hcal_{\textsc{Unf}}\gets\Hcal$
		\STATE $\Hcal\gets\{\text{\normalfont violated constraints returned by the ellipsoid method on } \circled{D} \,\, \text{\normalfont with objective } \gamma  \,\, \text{\normalfont and constraints } \Hcal_{\textsc{Unf}} \}$ 
		\UNTIL{\circled{D} is feasible with objective $\gamma$ (see Equation~\eqref{eq:feasibility_proj})}
		\STATE return $\Hcal_{\textsc{Unf}}$
	\end{algorithmic}	\label{alg:projection}
\end{algorithm}
At each iteration of the main loop, given an objective value $\gamma$, Algorithm~\ref{alg:projection} checks whether the problem \circled{D} is approximately feasible or unfeasible, by applying the ellipsoid algorithm with separation oracle $\Ocal_\alpha$. Let $\mathcal{H}$ be the set of constraints returned by the separation oracle (the separating hyperplanes due to the linear inequalities). At each iteration, the ellipsoid method is applied on the problem with explicit constraints in the current set $\Hcal_{\textsc{Unf}}$ (that is, each constraint in $\Hcal_{\textsc{Unf}}$ is explicitly checked for feasibility), while the other constraints are checked  through the approximate separation oracle. Algorithm~\ref{alg:projection} returns the set of violated constraints $\Hcal_{\textsc{Unf}}$ corresponding to the last value of $\gamma$ for which the problem was unfeasible.
Now, we describe how to implement the approximate separation oracle employed in Algorithm~\ref{alg:projection}. Then, we conclude the proof by showing how to build an approximate projection starting from the set $\Hcal_{\textsc{Unf}}$ computed as we just described.

\paragraph{Approximate Separation Oracle}
Let $(\bar \zvec, \bar \vvec, \bb{\bar \nu}, \bar \dvec)$ be a point in the space of dual variables. 
Then, let, for each $\theta\in\Theta$, $r\in \rec$, and $s\in\Scal_r$,  
\[
w^\theta_{r,s}\defeq \sum_{k\in s } \bar z_{r,s,k} \mu_\theta u^{r,k}_\theta.
\]

First, we can check in polynomial time if one of the constraint in $\Hcal$ is violated. If at least one of those constraint is violated, we output that constraint.
Moreover, if the constraint corresponding to the objective is violated, we can output a separation hyperplane in polynomial time since the constraint has a polynomial number of variables.
Then, by following the same rationale of the proof of Theorem~\ref{th:bayesian} (offline setting), we proceed with a case analysis in which we ensure it is possible to output a violated constraint when $|\nu_{\kvec}|$ or $|w^\theta_{r,s}|$ are too large to guarantee polynomial-time sovability by Definition~\ref{def:oracle}. 

\begin{itemize}
\item First, it has to hold $d_\theta \in [0,4 |K|]$ for each $\theta \in \Theta$.
Indeed, if $d_\theta<0$, then the constraint relative to $(\theta,\varnothing)$ would be violated. 
Otherwise, suppose that  there exists a $\theta$ with $\bar d_\theta > 4 |K|$. 
Two cases are possible: (i) the constraint corresponding to the objective is violated, which allows us to output a separation hyperplane; (ii) it holds 
\[
\sum_{\kvec \in K} \mleft(\bar \nu_{\kvec} y_{\kvec} - \frac{\bar \nu_{\kvec}^2}{4}\mright) > 4 |K|,
\] 
which implies that there exists a $\kvec \in K$ such that $\bar \nu_{\kvec} y_{\kvec}-\bar \nu_{\kvec}^2/4>4$. However, we reach a contradiction since, by assumption, $y_{\kvec}\le 2$ for each $\kvec \in K$, and therefore it must hold $\bar \nu_{\kvec} y_{\kvec}-\bar \nu_{\kvec}^2/4\le 2\bar\nu_{\kvec}-\bar\nu_{\kvec}^2/4\le 4$. 

\item Second, we show how to determine a violated constraint when  $\bar \nu_{\kvec} \notin [0,|K|+10]$.
Specifically, if there exists a $\kvec \in K$ for which $ \bar \nu_{\kvec} < 0$, then the objective is negative, and we can return a separation hyperplane (corresponding to Equation~\eqref{eq:feasibility_proj}). If there exists a $ \nu_{\kvec}>|K|+10$, then 
\begin{align*}
	\sum_{\kvec' \in K}  \mleft(\bar \nu_{\kvec'} y_{\kvec'}- \frac{\bar \nu_{\kvec'}^2}{4}\mright) & 
	\le	 2\nu_{\kvec} - \frac{\bar \nu_{\kvec}^2}{4}+\sum_{\kvec' \in K \setminus \{\kvec\}}  \mleft(2\bar \nu_{\kvec'}- \frac{\bar \nu_{\kvec'}^2}{4}\mright) \\
	& \le 2|K| + 20-\frac{|K|^2}{4}-5|K|-25+4|K|\\
	& =-\frac{|K|^2}{4}+|K|-5\\
	&< 0,
\end{align*}
where the first inequality follows by the assumption that $y_{\kvec}\le 2$ for each $\kvec \in K$, and the second inequality follows from the fact that  $2\nu_{\kvec} - \bar \nu_{\kvec}^2/4$ has its maximum in $\bar\nu_{\kvec}=4$ and, when $\bar\nu_{\kvec}\ge|K|+10$, the maximum is at $\bar\nu_{\kvec}=|K|+10$ since the function in concave.
Hence, we obtain that Constraint~\eqref{eq:feasibility_proj} is violated.

\item Finally, suppose that there exists a $\theta \in \Theta$, $r \in \rec$, $s \in \sset_r$ such that $w^\theta_{r,s} > 4 |K| $. Then, the constraint corresponding to $(\theta,s)$ is violated (because $d_\theta\le 4|K|$, otherwise we would have already determined a violated constraint in the first case of our analysis).
If, instead, there exists a $\theta \in \Theta$, $r \in \rec$, $s \in \sset_r$ such that $w^\theta_{r,s} < -4 |K| |\rec|-10$, then, for all the inequalities $(\theta,\svec')$ with $s'_r=s$, it holds $\bar d_\theta\ge 0$ and 
\[
\mu_\theta \sum_{\kvec \in K} \bar\nu_{\kvec} f_\theta(R_{\svec}^{\kvec}) +  \sum_{r' \in \rec\setminus \{r\}} w^\theta_{r',s'_{r'}} +w^\theta_{r,s'_r} \le 0.
\]
In this last case, all the inequalities corresponding to $(\theta,\svec')$ with $s'_r=s$ are guaranteed to be satisfied. Then, we can safely manage all the inequalities comprising of $w^\theta_{r,s}\le-4 |K| |\rec|-10$ by setting $w^\theta_{r,s}=-4 |K| |\rec|-10$.
\end{itemize}
After the previous steps, it is guaranteed that $|w^\theta_{r,s}| \le 4 |K| |\rec|+10$ for each $\theta,r,s$, and $\nu_{\kvec} \in [0,|K| +10]$ for each $\kvec$. 
Hence, we can employ an oracle $\Ocal_\alpha$ with $|w^\theta_{r,s}|$ and $\lambda^\theta_{\kvec}=\nu_{\kvec} \mu_\theta$, which is guaranteed to be polynomial in the size of the instance by Definition~\ref{def:oracle}. Let $\delta$ be an error parameter which will be defined in the remainder of the proof. For each $\theta \in \Theta$, we call the oracle $\Ocal_{\alpha}(\theta,K,\{\nu_{\kvec}\}_{\kvec \in K},\wvec^\theta,\delta)$. Each query to the oracle returns an $\svec^\theta$. If at least one of the constraints corresponding to a pair $(\theta,\svec^\theta)$ is violated, we output that constraint. Otherwise, if for each $\theta \in \Theta$ the constraint $(\theta,\svec^\theta)$  is satisfied, we conclude that the point is in the feasible region.

 \paragraph{Putting It All Together}
 Algorithm~\ref{alg:projection} terminates at objective $\gamma^\star$. It is easy to see that the algorithm terminates in polynomial time because it must return \emph{feasible} when $\gamma=0$.
 Our proof proceeds in two steps. 
 First, we prove that a particular problem obtained from \circled{P} has value at least $\gamma^\star$. Then, we prove that the solution of \circled{P} with only variables in $\Hcal_{\textsc{Unf}}$ has value close to $\gamma^\star$.
 Finally, we show that the two solutions are, respectively, the projection and an approximate projection on a set that includes $\alpha \X_{K}$. This will complete the proof.

If  the algorithm terminates at objective $\gamma^*$, the following convex optimization problem is feasible (see Theorem~\ref{th:bayesian}).\footnote{In the following, we will refer to the proof of Theorem~\ref{th:bayesian} when the steps of the two proofs are analogous.}
\begin{equation*}
\mleft\{\hspace{-1.25mm}\begin{array}{l}	 
\displaystyle \sum_{\kvec \in K} \mleft(\nu_{\kvec} y_{\kvec}- \nu_{\kvec}^2/4\mright) 	-\sum_{\theta\in\Theta} d_\theta \ge \gamma^\star \\[5mm]
\displaystyle d_\theta \ge \sum_{\kvec \in K} \nu_{\kvec} \mu_\theta  f_\theta(R_{\svec}^{\kvec})- \sum_{ r \in \rec,k\in s_r } z_{r,s_r,k} \,\mu_\theta u^{r,k}_\theta \hspace{1.5cm} \forall (\theta,\svec) \in \Hcal_{\textsc{Unf}}\\[5mm]
\displaystyle d_\theta \ge \sum_{\kvec \in K}  \alpha \nu_{\kvec} \mu_\theta  f_\theta(R_{\svec}^{\kvec})- \sum_{ r \in \rec,k\in s_r } z_{r,s_r,k}\, \mu_\theta u^{r,k}_\theta -\delta \hspace{.7cm}  \forall (\theta,\svec) \notin \Hcal_{\textsc{Unf}}.
\end{array}\mright.
\end{equation*} 
By strong duality, the following convex optimization problem has value at least $\gamma^\star$ 
\begin{equation*}
\circled{Pf}\mleft\{\hspace{-1.25mm}\begin{array}{l}	 
\displaystyle \min_{\phi,\xvec}\hspace{.2cm} \sum_{\kvec\in K} (x_{\kvec} -y_{\kvec})^2+ \delta \hspace{-.2cm}\sum_{(\theta,\svec) \notin \Hcal_{\textsc{Unf}}} \phi_{\theta}(\svec)\\ [6mm]
\textnormal{s.t. } \hspace{.5cm}\displaystyle \sum_{\theta\in\Theta}\mu_\theta \mleft(\sum_{\svec': s'_r=s} \,\phi_\theta(\svec') u^{r,k}_\theta\mright) \geq 0  \hspace{0.95cm} \forall r \in \rec,\forall s\in \sset_r ,\forall k \in \K_r:k\in s\\ [6mm]
\displaystyle \hspace{1cm}\sum_{\svec \in \sset}\phi_\theta(\svec)= 1 \hspace{3.8cm}\forall\theta\in\Theta \\ [5mm]
\displaystyle \hspace{1cm}\phi_{\theta}(\svec)\ge 0 \hspace{4.4cm}\forall \theta \in \Theta, \forall \svec \in \sset\\ [5mm]
\displaystyle \hspace{1cm}x_{\kvec} \le \sum_{\theta \in \Theta}\mleft( \sum_{\svec:(\theta,\svec) \in \Hcal_{\textsc{Unf}}}\hspace{-.2cm}\mu_\theta\,\phi_\theta(\svec) f_\theta(R_{\svec}^{\kvec})+ \alpha \hspace{-.3cm}\sum_{\svec:(\theta,\svec) \notin \Hcal_{\textsc{Unf}}}\hspace{-.3cm}\mu_\theta\,\phi_\theta(\svec) f_\theta(R_{\svec}^{\kvec})\mright)\hspace{.3cm}  \forall \kvec \in K.
\end{array}\mright.
\end{equation*} 
Moreover, since the algorithm did not terminate at value $\gamma^\star+\beta$, problem \circled{D} with value $\gamma^\star+\beta$ is unfeasible when restricting the set of constraints to $\Hcal_{\textsc{Unf}}$.
The primal problem \circled{P} restricted to primal variables corresponding to dual constraints in $\Hcal_{\textsc{Unf}}$ reads as follows
\begin{equation*}
\mleft\{\hspace{-1.25mm}\begin{array}{l}
\displaystyle \min_{\phi, \xvec} \hspace{.1cm} \sum_{\kvec \in K} (x_{\kvec}-y_{\kvec})^2\\[3mm]
\displaystyle \textnormal{s.t. }\hspace{.3cm} \sum_{\theta\in\Theta}\mu_\theta \mleft(\sum_{\substack{\svec:(\theta,\svec)\in \Hcal_{\textsc{Unf}},\\ s_r=s'}} \,\phi_\theta(\svec) u^{r,k}_\theta\mright) \geq 0  \hspace{1cm} \forall r \in \rec, s'\in \sset_r ,\forall k \in \K_r:k\in s'\\[8mm]
\displaystyle \hspace{.8cm} \sum_{\svec:(\theta,\svec)\in \Hcal_{\textsc{Unf}}}\phi_\theta(\svec)= 1 \hspace{3.4cm} \forall\theta\in\Theta \\[6mm]
\displaystyle \hspace{.8cm} \phi_{\theta}(\svec)\ge 0 \hspace{5.1cm} \forall (\theta,\svec)\in \Hcal_{\textsc{Unf}}\\[3mm]
\displaystyle \hspace{.8cm} x_{\kvec} \le  \sum_{\theta \in \Theta} \sum_{\svec:(\theta,\svec)\in \Hcal_{\textsc{Unf}}}\,\mu_\theta\,\phi_\theta(\svec) f_\theta(R_{\svec}^{\kvec}) \hspace{1.1cm} \forall \kvec \in K.
\end{array}\mright. 
\end{equation*}
By strong duality, the above problem has value at most $\gamma^\star+\beta$. Moreover, it has a polynomial number of variables and constraints because the ellipsoid method returns a set of constraints $\Hcal_{\textsc{Unf}}$ of polynomial size. Therefore, the above problem can be solved in polynomial time. 

%
A solution to the above problem is a feasible signaling scheme.
Let $(\xvec^\epsilon,\phi)$ be its solution. We have that $\xvec^\epsilon\in \bar \X_{K}$, with 
\[
\bar \X_{K}=\mleft\{\xvec: x_{\kvec}\le \sum_{\theta \in \Theta}\mleft( \sum_{\svec:(\theta,\svec) \in \Hcal_{\textsc{Unf}}}\hspace{-.2cm}\mu_\theta\,\phi_\theta(\svec) f_\theta(R_{\svec}^{\kvec})+ \alpha \hspace{-.3cm}\sum_{\svec:(\theta,\svec) \notin \Hcal_{\textsc{Unf}}}\hspace{-.3cm}\mu_\theta\,\phi_\theta(\svec) f_\theta(R_{\svec}^{\kvec})\mright)\hspace{.3cm}  \forall \kvec \in K, \phi \in \Phi\mright\}.
\]
It holds $\alpha \X_{K} \subseteq \bar \X_{K}$.
Now, we show that $\xvec^\epsilon$ is \emph{close} to $\xvec^\star$, where $\xvec^\star$ is the projection of $\yvec$ on $\bar \X_{K}$ (that is the solution of \circled{Pf} with $\delta=0$).
Since $x^\star$ is a feasible solution of \circled{Pf} and the minimum of  \circled{Pf} is at least $\gamma^\star$, it holds $||\xvec^\star-\yvec||^2+ \delta |\Theta|\ge \gamma^\star$.
Then, 
\begin{subequations}
	\begin{align*}
	||\xvec^\star-\yvec||^2+\delta|\Theta| +\beta & \ge
   \gamma^\star +\beta  \\	
	 &\ge ||\xvec^\epsilon-\yvec||^2 \\
	&= ||\xvec^\epsilon-\xvec^\star+\xvec^\star-\yvec||^2\\
	& =||\xvec^\epsilon-\xvec^\star||^2+||\xvec^\star-\yvec||^2 +2\langle	\xvec^\epsilon-\xvec^\star,\xvec^\star-\yvec\rangle\\ 
	 &\ge ||\xvec^\epsilon-\xvec^\star||^2 +||\xvec^\star-\yvec||^2,
	\end{align*}
	\end{subequations}
where the last inequality follows from $\langle	\xvec^\epsilon-\xvec^\star,\xvec^\star-\yvec\rangle\ge 0$, because $\xvec^\star$ is the projection of $\yvec$ on $\bar \X_{K}$ and $\xvec^\epsilon \in \bar \X_{K}$.
Hence, $||\xvec^\epsilon-\xvec^\star||^2 \le \delta |\Theta|+\beta$.
Finally, let $\xvec$ be a point in $\alpha \X_{K}$.
Then, 
\begin{align*}
||\xvec^\epsilon-\xvec||^2 & \le  ||\xvec^\epsilon-\xvec^\star||^2+ ||\xvec^\star-\xvec||^2\\
& \le ||\xvec^\epsilon-\xvec^\star||^2+  ||\yvec-\xvec||^2\\
& \le ||\yvec-\xvec||^2+\delta |\Theta|+\beta,
\end{align*}
 where the second inequality follow from the fact that $\xvec^\star$ is the projection of $\yvec$ on a superset of $\alpha \X_{K}$.
Setting $\delta=\frac{\epsilon}{2|\Theta|}$ and $\beta= \frac{\epsilon}{2}$ concludes the proof.
\end{proof}

\section{Proofs Omitted from Section~\ref{sec:submodular}}

In this section, we provide the complete proof of Theorem~\ref{thm:submodular}.

Firs, we introduce some preliminary, known results concerning the optimization over matroids.
Given a \emph{non-decreasing submodular} set function $f: 2^\G \to \mathbb{R}_+$ and a \emph{linear} set function $\ell: 2^\G \ni I \mapsto \sum_{i \in I} w_i$ defined for finite ground set $\G$ and weights $\wvec = (w_i)_{i \in \G}$ with $w_i \in \mathbb{R}$ for each $i \in \G$, let us consider the problem of maximizing the sum $f(I) + \ell(I)$ over the bases $I \in \B(\M)$ of a given matroid $\M \defeq (\G, \Ical)$.
We make use of a theorem due to~\citet{sviridenko2017optimal}, which, by letting $v_f \defeq \max_{I\in 2^\G} f(I)$, $v_\ell \defeq \max_{I\in 2^\G} | \ell(I)|$, and $v \defeq \max \{v_f , v_\ell \}$, reads as follows:
\begin{theorem}[Essentially Theorem~3.1 by~\citet{sviridenko2017optimal}]\label{thm:submodular}
	For every $\epsilon > 0$, there exists an algorithm running in time $\textnormal{poly}\left( |\G|,\frac{1}{\epsilon} \right)$ that produces a basis $I \in \B(\M)$ satisfying $f ( I ) + \ell ( I ) \ge \left( 1 - \frac{1}{e}  \right) f ( I' ) + \ell ( I' ) - O ( \epsilon ) v$ for every $I' \in \B(\M)$ with high probability.
\end{theorem}

Next, we provide the complete proof of Theorem~\ref{thm:submodular}.

\thmsubmodular*

\begin{proof}
	We show how to implement an approximation oracle $\Ocal_{\alpha}(\theta,K,\lambdavec,\wvec,\epsilon)$ (see Definition~\ref{def:oracle}) running in time $\text{\normalfont poly}\left( n,|K|,\max_{r,s} |w_{r,s}|,\max_{\kvec} \lambda_{\kvec} , \frac{1}{\epsilon} \right)$ for $\alpha = 1 - \frac{1}{e}$.
	Let $\M_\sset \defeq (\G_\sset, \Ical_\sset)$ be a matroid defined as in Section~\ref{sec:submodular} for direct signal profiles $\sset$.
	Let us recall that, given the relation between the bases of $\M_\sset$ and direct signals, each direct signal profiles $\svec \in \sset$ corresponds to a basis $I \in \B(\M_\sset)$, which is defined as $I \defeq \left\{ (r,s_r) \mid r \in \rec \right\}$.
	In the following, given a subset $I \subseteq \G_\sset$ and a type profile $\kvec \in K$, we let $R_I^{\kvec} \subseteq \rec$ be the set of receivers $r \in \rec$ such that there exits a pair $(r,s) \in I$ (for some signal $s \in \sset_r$) with the receiver's type $k_r$ being recommended to play $a_1$ under signal $s$; formally,
	\[
		R_I^{\kvec} \defeq \left\{ r \in \rec \mid  \exists (r,s) \in I : k_r \in s \right\}.
	\]

	First, we show that, when using matroid notation, the left-hand side of Equation~\eqref{eq:approximation_oracle} can be expressed as the sum of a non-decreasing submodular set function and a linear set function.
	To this end, let $f^{\lambdavec}_\theta:2^{\G_\sset} \to \mathbb{R}_+$ be defined as $ f^{\lambdavec}_\theta(I) = \sum_{\kvec \in K} \lambda_{\kvec}  f_\theta(R_I^{\kvec} )$ for every subset $I \subseteq \G_\sset$.
	We prove that $f^{\lambdavec}_\theta$ is submodular.
	Since $f^{\lambdavec}_\theta$ is a suitably defined weighted sum of the functions $f_\theta$, it is sufficient to prove that, for each type profile $\kvec \in K$, the function $f_\theta: 2^\rec \to [0,1]$ is submodular in the sets $R_I^{\kvec}$.
	For every pair of subsets $I \subseteq I' \subseteq \G_\sset$, and for every receiver $r \in \rec$ and signal $s \in \sset_r$, the marginal contribution to the value of function $f_\theta$ due to the addition of element $(r,s)$ to the set $I$ is:
	\begin{align*}
		f_\theta( R_{I \cup (r,s)}^{\kvec}   )-f_\theta( R_I^{\kvec}   ) & = \event \left\{  k_r \in s \wedge \nexists (r,s') \in I :  k_r \in s' \right\} \Big( f_\theta( R_I^{\kvec} \cup \{ r\} )-f_\theta( R_I^{\kvec}   ) \Big)  \ge \\
		& \geq \event \left\{  k_r \in s \wedge \nexists (r,s') \in I' :  k_r \in s' \right\} \Big( f_\theta( R_I^{\kvec} \cup \{ r\} )-f_\theta( R_I^{\kvec}   ) \Big) \ge\\
		& \ge	\event \left\{  k_r \in s \wedge \nexists (r,s') \in I' :  k_r \in s' \right\}  \Big( f_\theta( R_{I'}^{\kvec} \cup \{ r\} )-f_\theta( R_{I'}^{\kvec}   ) \Big) =\\
		& =f_\theta( R_{I' \cup (r,s)}^{\kvec}   )-f_\theta( R_{I'}^{\kvec}   ),
	\end{align*}
	where the last inequality holds since the functions $f_\theta$ are submodular by assumption.
	Since the last expression is the marginal contribution to the value of function $f_\theta$ due to the addition of element $(r,s)$ to the set $I'$, the relations above prove that the function $f^{\lambdavec}_\theta$ is submodular.
	Let $\ell^{\wvec} : 2^{\G_\sset} \to \mathbb{R}_+$ be a linear function such that $\ell^{\wvec}(I) = \sum_{r\in\rec} w_{r,s_r}$ for every basis $I \subseteq \B(\M_\sset)$, with each $s_r \in \sset_r$ being the signal of receiver $r \in \rec$ specified by the signal profile corresponding to the basis, namely $(r,s_r) \in I$.
	Then, we have that finding a signal profile $\svec \in \sset$ satisfying Equation~\eqref{eq:approximation_oracle} is equivalent to finding a basis $I \in \B(\M_\sset)$ of the matroid $\M_\sset$ (representing a direct signal profile) such that:
	\[
		f^{\lambdavec}_\theta(I)+ \ell^{\wvec}(I) \ge  \max_{I^\star \in \B(\M_\sset)} \left\{\alpha \sum_{\kvec \in K} f^{\lambdavec}_\theta(I^\star) + \ell^{\wvec}(I^\star) \right\}-\epsilon.
	\]
	Notice that, for $\epsilon' > 0$, the algorithm of Theorem~\ref{thm:submodular} by~\citet{sviridenko2017optimal} can be employed to find a basis $I \in \B(\M_\sset)$ such that $f^{\lambdavec}_\theta ( I ) + \ell^{\wvec} ( I ) \ge \left( 1 - \frac{1}{e}  \right) f^{\lambdavec}_\theta ( I' ) + \ell^{\wvec} ( I' ) - O ( \epsilon' ) v$ for every $I' \in \B(\M)$ with high probability, employing time polynomial in $|\G_\sset|$ and $\frac{1}{\epsilon}$.
	Since $|\G_\sset|$ is polynomial in $n$ and $v$ is  polynomial in $|K| ,\max_{r,s} |w_{r,s}|$ and $\max_{\kvec} \lambda_{\kvec}$, by setting $\epsilon' = O(\frac{\epsilon}{v})$ and $\alpha = 1 - \frac{1}{e}$, we get the result.
\end{proof}

\end{document}